\documentclass[11pt,a4paper]{article}
\usepackage{times}
\usepackage{a4wide}

\usepackage{latexsym}
\usepackage{amssymb}
\usepackage{epsfig}
\usepackage{amsfonts}
\usepackage{color}
\usepackage{amsmath}
\usepackage{fleqn}
\usepackage[T1]{fontenc}
\usepackage{ifthen}
\usepackage{amsthm}
\usepackage{tkz-berge}

\newboolean{commentson} % to control comments
\setboolean{commentson}{true}
\newcommand{\comment}[1]
{\ifthenelse{\boolean{commentson}}
   {{\par\noindent\mbox{}{\small[ *** #1 ]\par}\noindent\par}}{}}

\begin{document}

\newcommand{\bull}{\rule{.85ex}{1ex} \par \bigskip}
\newenvironment{sketch}{\noindent {\bf Proof (sketch):\ }}{\hfill \bull}

\newtheorem{theorem}{Theorem}[section]
\newtheorem{definition}[theorem]{Definition}
\newtheorem{proposition}[theorem]{Proposition}
\newtheorem{lemma}[theorem]{Lemma}
\newtheorem{corollary}[theorem]{Corollary}
\newtheorem{conjecture}[theorem]{Conjecture}
\newtheorem{exmp}{Example}
\newtheorem{notation}[theorem]{Notation}
\newtheorem{problem}{Problem}
\newtheorem{remark}[theorem]{Remark}
\newtheorem{observation}[theorem]{Observation}

\newcommand{\QCSP}[1]{\mbox{\rm QCSP$(#1)$}}
\newcommand{\CSP}[1]{\mbox{\rm CSP$(#1)$}}
\newcommand{\MCSP}[1]{\mbox{{\sc Max CSP}$(#1)$}}
\newcommand{\wMCSP}[1]{\mbox{\rm weighted Max CSP$(#1)$}}
\newcommand{\cMCSP}[1]{\mbox{\rm cw-Max CSP$(#1)$}}
\newcommand{\tMCSP}[1]{\mbox{\rm tw-Max CSP$(#1)$}}
\renewcommand{\P}{\mbox{\bf P}}
\newcommand{\G}[1]{\mbox{\rm I$(#1)$}}
\newcommand{\NE}[1]{\mbox{$\neq_{#1}$}}

\newcommand{\MCol}[1]{\mbox{\sc Max $#1$-Col}}

\newcommand{\NP}{\mbox{\bf NP}}
\newcommand{\NL}{\mbox{\bf NL}}
\newcommand{\PO}{\mbox{\bf PO}}
\newcommand{\NPO}{\mbox{\bf NPO}}
\newcommand{\APX}{\mbox{\bf APX}}
\newcommand{\Aut}{\mbox{\rm Aut}^*}
\newcommand{\bound}{\mbox{\rm -$B$}}

\newcommand{\GIF}[3]{\ensuremath{h_\{{#2},{#3}\}^{#1}}}

\newcommand{\Spmod}{\mbox{\rm Spmod}}
\newcommand{\Sbmod}{\mbox{\rm Sbmod}}

\newcommand{\Inv}[1]{\mbox{\rm Inv($#1$)}}
\newcommand{\Pol}[1]{\mbox{\rm Pol($#1$)}}
\newcommand{\sPol}[1]{\mbox{\rm s-Pol($#1$)}}

\newcommand{\un}{\underline}
\newcommand{\ov}{\overline}
\def\ar{\hbox{ar}}
\def\vect#1#2{#1 _1\zdots #1 _{#2}}
\def\zd{,\ldots,}
\let\sse=\subseteq
\let\la=\langle
\def\lla{\langle\langle}
\let\ra=\rangle
\def\rra{\rangle\rangle}
\let\vr=\varrho
\def\vct#1#2{#1 _1\zd #1 _{#2}}
\newcommand{\va}{{\bf a}}
\newcommand{\vb}{{\bf b}}
\newcommand{\vc}{{\bf c}}
\newcommand{\bx}{{\bf x}}
\newcommand{\by}{{\bf y}}
\def\Z{{\bur Z^+}}
\def\R{{\bur R}}
\def\D{{\cal D}}
\def\F{{\cal F}}
\def\I{{\cal I}}
\def\C{{\cal C}}
\def\U{{\cal U}}
\def\K{{\cal K}}
\def\Lat{{\cal L}}

\def\2mat#1#2#3#4#5#6#7#8{
\begin{array}{c|cc}
$~$ & #3 & #4\\
\hline
#1 & #5& #6\\
#2 & #7 & #8 \end{array}}

\renewcommand{\phi}{\varphi}
\renewcommand{\epsilon}{\varepsilon}

\def\tup#1{\mathchoice{\mbox{\boldmath$\displaystyle#1$}}
{\mbox{\boldmath$\textstyle#1$}}
{\mbox{\boldmath$\scriptstyle#1$}}
{\mbox{\boldmath$\scriptscriptstyle#1$}}}
\newcommand{\draft}{\begin{center}\huge Draft!!! \end{center}}
\newcommand{\void}{\makebox[0mm]{}}     % Often needed to fool LaTeX

%%%%%%%%%%%%%%%%%%%%%%%%%%%%%%%%%%%%%%%%%
%% Commands to write text in math mode %%
%%%%%%%%%%%%%%%%%%%%%%%%%%%%%%%%%%%%%%%%%

\renewcommand{\text}[1]{\mbox{\rm \,#1\,}}        % To write text in mathmode

%%%%%%%%%%
%% Math %%
%%%%%%%%%%

%% Sets
\renewcommand{\emptyset}{\varnothing}  % The usual empty set,  not LaTeX std.
                                   % Requires the ``amssymbols'' option
\newcommand{\union}{\cup}               % Union
\newcommand{\intersect}{\cap}           % intersection
\newcommand{\setdiff}{-}                % Set difference
\newcommand{\compl}[1]{\overline{#1}}   % Set complement
\newcommand{\card}[1]{{|#1|}}           % Cardinality
\newcommand{\set}[1]{\{{#1}\}} % Set
\newcommand{\st}{\ |\ }                 % `such that' symbol in sets
\newcommand{\suchthat}{\st}             % Old, do not use
\newcommand{\cprod}{\times}             % Cartesian product
\newcommand{\powerset}[1]{{\bf 2}^{#1}} % Powerset

%% Other math
\newcommand{\tuple}[1]{\langle{#1}\rangle}  % Tuple
\newcommand{\seq}[1]{\langle #1 \rangle}
\newcommand{\emptyseq}{\seq{}}
\newcommand{\floor}[1]{\left\lfloor{#1}\right\rfloor}
\newcommand{\ceiling}[1]{\left\lceil{#1}\right\rceil}

%% Functions
\newcommand{\map}{\rightarrow}
\newcommand{\fncomp}{\!\circ\!}         % Function composition

%% Relations
\newcommand{\transclos}[1]{#1^+}
\newcommand{\reduction}[1]{#1^-}        % Warning! non-standard notation

%%%%%%%%%%%%%%%%%%%%%%%
%%% Macros for text %%%
%%%%%%%%%%%%%%%%%%%%%%%

\newcommand{\perfimp}{\stackrel{p}{\Longrightarrow}}

\newcommand{\ie}{{\em ie.}}                % Self explaining
\newcommand{\eg}{{\em eg.}}
\newcommand{\paper}{paper}                % -"-

\newcommand{\emdef}{\em}                   % Font for defined terms
\newcommand{\rinterpretation}{${\mathbb R}$-interpretation}
\newcommand{\rmodel}{${\mathbb R}$-model}
\newcommand{\transp}{^{\rm T}}

\newcommand{\unprint}[1]{}
\newcommand{\blankline}{$\:$}

\newcommand{\prob}[1]{{\sc #1}}

\newcommand{\Solv}{{\it TSolve}}
\newcommand{\Neg}{{\it Neg}}
\newcommand{\logname}{XX}

\newcommand{\props}{{\it props}}
\newcommand{\rels}{{\it rels}}
\newcommand{\deduce}{\vdash_p}

\newcommand{\pform}{{\rm Pr}}
\newcommand{\axform}{{\rm AX}}
\newcommand{\axset}{{\bf AX}}
\newcommand{\resdeduce}{\vdash_{\rm R}}
\newcommand{\resaxdeduce}{\vdash_{\rm R,A}}

\newcommand{\cmis}{{\em \#mis}}
\newcommand{\combine}{{\em comb}}

\newcommand{\xcsp}{{\sc X-Csp}}
\newcommand{\csp}{{\sc Csp}}

\newcommand{\cc}[1]{\textnormal{\textbf{#1}}} % Complexity Class
\newcommand{\opt}[0]{\textrm{{\sc opt}}}

\newcommand{\MC}{mc}
\newcommand{\Ha}{\textrm{\textit{H\aa}}}
\renewcommand{\atop}[2]{\genfrac{}{}{0pt}{}{#1}{#2}}
\newcommand{\GHAT}{{\cal G_\equiv}}
\newcommand{\HOMEQ}{\equiv}

\pagestyle{plain}

\author{Robert Engstr\"om\footnote{\tt{engro910@student.liu.se}}, Tommy F\"arnqvist\footnote{\tt{\{tomfa, petej, johth\}@ida.liu.se }}, Peter Jonsson\footnotemark[2],  and Johan Thapper\footnotemark[2]\\
\\
\small
Department of Computer and Information Science\\
\small
Link\"{o}pings universitet\\
\small
SE-581 83 Link\"{o}ping, Sweden\\}

\title{Graph Homomorphisms, Circular Colouring, and \\ Fractional Covering by $H$-cuts}

\date{}
\maketitle
\bibliographystyle{abbrv}

\begin{abstract}
A graph homomorphism is a vertex map which carries edges from a source 
graph to edges in a
target graph. The instances of the \emph{Weighted Maximum 
$H$-Colourable Subgraph} problem ($\MCol{H}$) are
edge-weighted graphs $G$ and
the objective is to find a subgraph of $G$ that has maximal total edge
weight, under 
the condition that the subgraph has a homomorphism to $H$; note that for
$H=K_k$ this 
problem is equivalent to {\sc Max $k$-cut}. F\"arnqvist et al.\ have
introduced a parameter on the space of graphs that allows close study of
the approximability properties of $\MCol{H}$. Specifically, it can be
used to extend previously known (in)approximability results to larger
classes of graphs. Here, we investigate the properties of this parameter
on circular complete graphs $K_{p/q}$, where $2 \leq p/q \leq 3$. The
results are extended to $K_4$-minor-free graphs and graphs with bounded
maximum average degree. We also consider connections with
\v{S}\'{a}mal's work on fractional covering by cuts: we address,
and decide, two conjectures concerning cubical chromatic numbers.

\end{abstract}

\section{Introduction}
Denote by ${\cal G}$ the set of all simple, undirected and
finite graphs. 
A \emph{graph homomorphism} from $G \in {\cal G}$ to $H \in {\cal G}$ is a vertex map which
carries the edges in $G$ to edges in $H$.
The existence of such a map will be denoted by $G \rightarrow H$.
%If both $G \rightarrow H$ and $H \rightarrow G$, the graphs $G$ and $H$
%are said to be \emph{homomorphically equivalent} (denoted $G \HOMEQ H$).
For a graph $G \in {\cal G}$, let
${\cal W}(G)$ be the set of \emph{weight functions}
$w : E(G) \rightarrow {\mathbb Q}^+$ assigning weights
to edges of $G$.
%For a $w \in {\cal W}(G)$, we let
%$\|w\| = \sum_{e \in E(G)} w(e)$ denote the total weight of $G$.
Now,
  {\em Weighted Maximum $H$-Colourable Subgraph} (\MCol{H}) is the
maximisation problem with
  \begin{description}
  \item[Instance:] An edge-weighted graph $(G,w)$, where $G \in {\cal
G}$ and
    $w \in {\cal W}(G)$.
  \item[Solution:] A subgraph $G'$ of $G$ such that $G' \rightarrow H$.
  \item[Measure:] The weight of $G'$ with respect to $w$.
  \end{description}

\noindent
Given an edge-weighted graph $(G,w)$, denote by $\MC_H(G,w)$ the measure
of the optimal solution to the problem \MCol{H}.
Denote by $mc_k(G,w)$ the
(weighted) size of a largest $k$-cut in $(G,w)$.
This notation is justified by the fact that
$\MC_k(G,w) = \MC_{K_k}(G,w)$.
In this sense, \MCol{H} generalises {\sc Max $k$-cut} which is a
well-known and well-studied problem that is computationally hard
when $k > 1$.
Since \MCol{H} is a hard problem to solve exactly, efforts have been
made to find suitable approximation algorithms.
F\"arnqvist et al.~\cite{farnqvist:etal:09} introduce a
method that can be used to extend previously known
(in)approximability bounds on \MCol{H} to new and larger classes of
graphs. For example, they 
present concrete approximation ratios for certain graphs (such as the odd cycles) 
and
near-optimal asymptotic results for large graph classes.
The fundament of this promising technique is the ability to compute 
(or closely approximate) a function
$s: {\cal G} \times {\cal G} \rightarrow {\mathbb R}$ defined as follows: 
\begin{equation}
\label{eq:s}
s(M,N) = \inf_{\substack{G \in {\cal G} \\ \omega \in {\cal
W}(G)}}{\frac{mc_M(G,\omega)}{mc_N(G,\omega)}}.
\end{equation}

It is not surprising that estimating $s(M,N)$ 
is, in many cases,
non-trivial. One way is to solve a certain linear program that
we present in Section~\ref{sec:linprog}: the program
can be tedious to write down since it is based on the structure
of $N$'s automorphism group, and can be prohibitively large.
Another way is to use the following lemma:

\begin{lemma}[\cite{farnqvist:etal:09}]
\label{lem:sandwich}
Let $M \rightarrow H \rightarrow N$. Then, $s(M,H) \geq s(M,N)$ and
$s(H,N) \geq s(M,N)$.
\end{lemma}

It is apparent that in order to use this result effectively, we need a large
selection of graphs $M,N$ that are known to be close to each other with respect
to $s$. For the moment, the set of such examples is quite meagre.
Hence, we set out to investigate how the function $s$ behaves on
certain classes of graphs. In Section~\ref{sec:meas}, we will take a careful
look at 3-colourable circular complete graphs and, amongst other things,
find that $s$ is constant between a large number of these graphs.
Moreover, we will extend bounds on $s$ to other classes of graphs using
known results about homomorphisms to circular complete graphs; 
examples include $K_4$-minor-free graphs and graphs with bounded
maximum average degree. 

Yet another way of estimating the function $s$ is to relate it
to other graph parameters. In this vein, Section~\ref{sec:cut} is dedicated
to generalising the work of \v{S}\'{a}mal~\cite{samal:05,samal:06}
on fractional covering by cuts to obtain a new family of `chromatic
numbers'. This reveals that $s(M,N)$ and the new chromatic numbers
$\chi_M(N)$ are closely related quantities, which provides us with an
alternative way of computing $s$.
We also use our knowledge about the behaviour of $s$ to
disprove a conjecture by \v{S}\'{a}mal concerning the cubical chromatic
number and, finally, we decide in the positive another conjecture by
\v{S}\'{a}mal concerning the same parameter.
We conclude the paper, in Section~\ref{sec:open}, by discussing
open problems and directions for future research.
To improve readability some proofs are deferred to the appendices.

% *****  H borde det finnas en karta er resten av pappret. ********

\section{A Linear Program for $s$}
\label{sec:linprog}

F\"arnqvist et al.~\cite{farnqvist:etal:09} have identified 
an alternative expression for $s(M,N)$
which depends on the automorphism group of $N$.
Let $M$ and $N \in {\cal G}$ be graphs and let $A = \Aut(N)$ be the
(edge) automorphism group of $N$, i.e., $\pi \in A$ acts on $E(N)$
by permuting the edges.
Let $\hat{{\cal W}}(N)$ be the set of
weight functions $\omega \in {\cal W}(N)$ which satisfy 
$\sum_{e \in E(N)} \omega(e) = 1$
and for which $\omega(e) = \omega(\pi \cdot e)$ for all $e \in E(N)$ and
$\pi \in \Aut(N)$. That is, the weight functions in $\hat{{\cal W}}(N)$
are constant over the edges belonging to each orbit of $\Aut(N)$.
\begin{lemma}[\cite{farnqvist:etal:09}]
\label{lem:auto}
   Let $M,N\in {\cal G}$.
    Then,
    $s(M,N) = \inf_{w \in {\cal {\hat W}}(N)} mc_M(N,w)$.
    In particular, when $N$ is edge-transitive,
    $s(M,N) = mc_M(N,1/|E(N)|)$.
\end{lemma}
Lemma~\ref{lem:auto} shows that in order to determine $s(M,N)$, it is
sufficient to minimise $mc_M(N,\omega)$ over $\hat{{\cal W}}(N)$, and it follows
that $s(M,N)$ can be computed by solving a linear program.
For $i \in \{1,\ldots,r\}$, let $A_i$ be the orbits
of $\Aut(N)$ and, for $f : V(N) \rightarrow V(M)$, define
\begin{equation}
f_i = | \{u v \in A_i \;|\; f(u) f(v) \in E(M)\}|.
\end{equation}
That is, $f_i$ is the number of edges in $A_i$ which are mapped to an
edge in $M$ by $f$. The measure of a solution $f$ when $\omega \in
\hat{{\cal W}}(N)$ is equal to $\sum_{i=1}^{r}{\omega_i \cdot f_i}$
where $\omega_i$ is the weight of an edge in $A_i$. Given an $\omega$,
the measure of a solution $f$ depends only on the vector
$(f_1,\ldots,f_r) \in {\mathbb N}^r$. We call this vector the {\em
signature} of $f$. When there is no risk of confusion,
we will let $f$ denote the signature as well.
Since we have seen that the measure of a solution only
depends on its signature the solution space is taken to be the set of
possible signatures  
\begin{equation}
F = \{f \in {\mathbb N}^r \,|\, f\mbox{ is a signature of a solution to $(N,\omega)$ of
$\MCol{M}$} \}.
\end{equation}
The variables of the linear program are $\omega_1,\ldots,\omega_r$ and
$s$, where $\omega_i$ represents the weight of each element in the orbit
$A_i$ and $s$ is an upper bound on the signatures measure.
\begin{equation*}
\tag{LP}
\label{lp}
\begin{array}{ll}
    \min s \\
    \sum_i f_i \cdot \omega_i \leq s & \text{for each $(f_1, \ldots, f_r) \in
F$} \\
    \sum_i |A_i| \cdot \omega_i = 1 & \ \text{and} \ \omega_i, s \geq 0\\
\end{array}
% \begin{split}
% \mbox{min } & s \\
% \mbox{s.t. } & \sum_i{f_i \cdot \omega_i} \leq s  \quad \mbox{ for each } f \in F \\
% & \sum_i{|A_i|\cdot \omega_i} = 1 \\
% & \omega_i,s \geq 0
% \end{split}
\end{equation*}
Given a solution $\omega_i,s$ to this program, $\omega(e) = \omega_i$
when $e \in A_i$ is a weight function which minimises $mc_M(G,\omega)$.
The value of this solution is $s = s(M,N)$.

\section{Solutions to (\ref{lp}) for Circular Complete Graphs} \label{sec:meas}

A circular complete graph $K_{p/q}$ is a graph with vertex set
$\{v_0,v_1,\ldots,v_{n-1}\}$ and edge set $E(K_{p/q}) = \{v_i v_j \ | \ q \leq |i-j|
\leq p-q\}$. This can be seen as placing the vertices on a circle and
connecting two vertices by an edge if they are at a distance at least
$q$ from each other. A fundamental property of these graphs is that
$K_{p/q} \rightarrow K_{p'/q'}$ iff $p/q \leq p'/q'$.
Due to this fact, when we write $K_{p/q}$, we will assume that $p$ and $q$
are relatively prime.
We will denote the orbits of the action of $\Aut(K_{p/q})$ by
$A_c = \{ v_i v_j \in E(K_{p/q}) \;|\; j-i \equiv q+c-1 \text{ (mod $p$)} \}$,
for $c = 1, \ldots, \lceil \frac{p-2q+1}{2} \rceil$.
We finally note that a
homomorphism from a graph $G$ to $K_{p/q}$ is called a (circular)
$(p/q)$-colouring of $G$. 
More information on this topic can be gained
from the book by Hell and Ne\v{s}et\v{r}il~\cite{HN04} and from the
survey by Zhu~\cite{zhu:survey}.

In this section we start out by investigating $s(K_r, K_t)$
for rational numbers $2 \leq r < t \leq 3$.
In Section~\ref{ssec:k2}, we fix $r=2$ and choose $t$ so that
$\Aut(K_{t})$ has few orbits.
We find some interesting properties of these numbers which lead
us look at the case $r = 2+1/k$ in Section~\ref{ssec:odd}.
Our approach is based on relaxing the linear program (\ref{lp})
that was presented in Section~\ref{sec:linprog}, combined with
arguments that our chosen relaxations in fact find the optimum
in the original program.

\subsection{Maps to $K_2$} \label{ssec:k2}

We consider $s(K_2, K_{t})$ for
$t = 2 + n/k$ with $k > n \geq 1$, where $n$ and $k$ are integers.
The number of orbits of $\Aut(K_{t})$ then equals $\lceil (n+1)/2 \rceil$.
%In particular, this number is solely determined by $n$.
We choose to begin our study of $s(K_2, K_{t})$ using small values of $n$.
When $n = 1$, $K_{2+1/k}$ is isomorphic to the cycle $C_{2k+1}$.
The value of $s(K_2,C_{2k+1}) = 2k/(2k+1)$, for $k \geq 1$ 
was obtained in~\cite{farnqvist:etal:09}.
Combined with the following result, where we set $t = 2+2/(2k-1) = \frac{4k}{2k-1}$, 
this has an immediate and perhaps surprising consequence.
\begin{proposition}
\label{prop:4k+4}
Let $k\geq 1$ be an integer, then $s(K_2,K_{\frac{4k}{2k-1}})=\frac{2k}{2k+1}$.
\end{proposition}
\begin{proof}
Let $V(K_{\frac{4k}{2k-1}})= \{v_0,v_1,\ldots,v_{4k-1}\}$ and $V(K_2) =\{w_0,w_1\}$.
We will present two maps $f, h : V(K_{\frac{4k}{2k-1}}) \rightarrow V(K_2)$.
$f$ sends a vertex $v_i$ to $w_0$ if $0 \leq i < 2k$ and to $w_1$ if
$2k \leq i < 4k$.
It is not hard to see that $f = (4k-2, 2k)$.
The map $h$ sends $v_i$ to $w_0$ if $i$ is even and to $w_1$ if $i$ is odd.
Then, $h$ maps all of $A_1$ to $K_2$ but none of the edges in $A_2$,
so $h = (4k,0)$.
It remains to argue that these two solutions suffice to determine $s$.
But we see that any map $g$ with $g_2 > 0$ must cut at least two edges in
the even cycle $A_1$, leading to $g_1 \leq 4k-2$, thus $g \leq f$,
componentwise.
The proposition now follows by solving the relaxation of (\ref{lp})
using only the two inequalities obtained from $f$ and $h$.
\end{proof}

%\begin{proof}
%The solution to~(\ref{lp}) with $\omega_1,\omega_2,s$ as basic variables gives $\omega_1 = \frac{1}{4k+2}$, $\omega_2 = \frac{1}{k(4k+2)}$ and $s = \frac{2k}{2k+1}$. With the reduced costs being ($\frac{2k}{2k+1}, \frac{1}{2k+1})$, hence this is the optimal soultion.
%\end{proof}
\begin{corollary}
\label{cor:intervals}
Let $k \geq 1$ and $2 \leq r < \frac{2k+1}{k} \leq t \leq \frac{4k}{2k-1}$.
Then, $s(K_r,K_{t}) = \frac{2k}{2k+1}$.
\end{corollary}
\begin{proof}
  Note that we have the chain of homomorphisms
  $K_2 \rightarrow K_r \rightarrow K_{\frac{2k+1}{k}} \rightarrow K_t \rightarrow K_{\frac{4k}{2k-1}}$.
  By Lemma~\ref{lem:sandwich}, we get $s(K_r,K_{\frac{2k+1}{k}}) \geq s(K_2,K_{\frac{2k+1}{k}}) = \frac{2k}{2k+1}$. But since $K_{\frac{2k+1}{k}} \not\rightarrow K_r$, and $K_{\frac{2k+1}{k}}$ is edge-transitive with $2k+1$ edges, $s(K_r, K_{\frac{2k+1}{k}}) \leq \frac{2k}{2k+1}$ and therefore $s(K_r,K_{\frac{2k+1}{k}}) =  \frac{2k}{2k+1}$. 
  Again by Lemma~\ref{lem:sandwich}, we have
  $
  \frac{2k}{2k+1} = s(K_r, K_{\frac{2k+1}{k}}) \geq s(K_r, K_{t}) \geq s(K_2, K_{\frac{4k}{2k-1}}) = \frac{2k}{2k+1}.
  $
\end{proof}

We find that there are intervals $I_k = \{ t \in \mathbb{Q} \;|\; 2+1/k \leq t \leq 2+2/(2k-1) \}$ where $s(t) = s(K_r, K_{t})$ is constant.
In Figure~\ref{fig:interval} these intervals are shown for the first few values of $k$. The intervals $I_k$ form an infinite sequence with endpoints tending to $2$. 
Similar intervals appear throughout the space of circular complete graphs.
More specifically,  F\"{a}rnqvist et al.~\cite{farnqvist:etal:09} have shown
that $s(K_n,K_{2m-1}) = s(K_n,K_{2m})$ for arbitrary integers $n, m \geq 2$.
Furthermore, it can be proved that $s(K_2, K_n) = s(K_{8/3}, K_n)$ for $n \geq 3$.
Two applications of Lemma~\ref{lem:sandwich} now shows that $s(K_r, K_{t})$ is 
constant on the regions $[2, 8/3] \times J_m$, where
  $J_m = \{ t \in \mathbb{Q} \;|\; 2m-1 \leq t \leq 2m \}$.

\begin{figure}[h]
\centering
\begin{tikzpicture}
\draw[dashed] (0,0) -- (1.875,0);
\draw[very thick] (1.875,0) -- (2.1428,0);
\draw[thin] (2.1428,0) -- (2.5,0);
\draw[very thick] (2.5,0) -- (3,0);
\draw[thin] (3,0) -- (3.75,0);
\draw[very thick] (3.75,0) -- (5,0);
\draw[thin] (5,0) -- (7.5,0);
\draw[thick](0,-0.2) -- (0,0.2);
\draw[thick](1.875,-0.2) -- (1.875,0.2);
\draw[thick](2.1428,-0.2) -- (2.1428,0.2);
\draw[thick](2.5,-0.2) -- (2.5,0.2);
\draw[thick](3,-0.2) -- (3,0.2);
\draw[thick](3.75,-0.2) -- (3.75,0.2);
\draw[thick](5,-0.2) -- (5,0.2);
\draw[thick](7.5,-0.2) -- (7.5,0.2);
\node at (0,0.5) {$2$};
\node at (1.875,0.5) {$\frac{9}{4}$};
\node at (2.1428,0.5) {$\frac{16}{7}$};
\node at (2.5,0.5) {$\frac{7}{3}$};
\node at (3,0.5) {$\frac{12}{5}$};
\node at (3.75,0.5) {$\frac{5}{2}$};
\node at (5,0.5) {$\frac{8}{3}$};
\node at (7.5,0.5) {$3$};
\node at (0,-0.5) {$s(K_2,K_r)=$};
\node at (2,-0.5) {$\frac{8}{9}$};
\node at (2.75,-0.5) {$\frac{6}{7}$};
\node at (4.375,-0.5) {$\frac{4}{5}$};
\node at (-0.45,0.5) {$r=$};
\end{tikzpicture}
\caption{The space between $2$ and $3$ with the intervals $I_k$ marked for $k = 2, 3, 4$.}\label{fig:interval}
\end{figure}
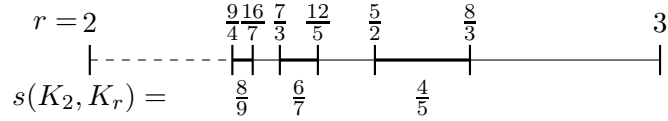

As we proceed with determining $s(K_2, K_{t})$ we can now, 
thanks to Corollary~\ref{cor:intervals}, disregard those $t$ which fall
inside these constant intervals.
%We consider $r$ of the form $2 + n/k$, for a fixed $n$ and note that $n = 1$
%corresponds to odd cycles and $n = 2$ was treated in 
%Proposition~\ref{prop:4k+4}.
For $t = 2 + 3/k$, we see that if
$k \equiv 0 $ (mod 3), then $r$ is an odd cycle,
and if $k \equiv 2 $ (mod 3), then $t \in I_{k+1}$.
Therefore, we assume that $t$ is of the form 
$2+3/(3k+1) = \frac{6k+5}{3k+1}$ for an integer $k \geq 1$.
%The action of the automorphism group on $K_r$ has two orbits, each of size $6k+5$.

\begin{proposition} \label{prop:3k+1}
Let $k \geq 1$ be an integer. Then, $s(K_2,K_{\frac{6k+5}{3k+1}})=\frac{6k^2+8k+3}{6k^2+11k+5}=1 - \frac{3k+2}{(k+1)(6k+5)}$.
\end{proposition}

For $t = 2 + 4/k$, we find that we only need to consider the case when $k \equiv 1 \mbox{ (mod 4) }$.
We then have graphs $K_{t}$ with $t = 2+4/(4k+1) = \frac{8k+6}{4k+1}$ for integers $k \geq 1$.
%For the following result, two solutions with signatures
%$(8k+4,4k+4,2k+1)$ and $(8k+2,8k+4,4k+3)$ can be found and we solve
%the relaxation of~(\ref{lp}) using only these two.
%Then, we can show that there is no signature which makes the obtained solution
%infeasible in the original program.

%To prove the optimality of the relaxation of~(\ref{lp}), we can show that When determining $s(K_2,K_{\frac{8k+6}{4k+1}})$ the process will be a little bit different than for the previous ones. First two solution signatures will be presented and a relaxed version of~(\ref{lp}) will be solved. An optimal solution to a relaxed linear program is an optimal solution to the original linear program if it is a feasible solution there. So it will then be shown that there exists no signature which makes this solution infeasible in~(\ref{lp}), and thus the two presented signatures forms a complete set.

\begin{proposition} \label{prop:4k+1}
Let $k \geq 1$ be an integer. Then, $s(K_2,K_{\frac{8k+6}{4k+1}})=\frac{8k^2+6k+2}{8k^2+10k+3}=1 - \frac{4k+1}{(k+1/2)(8k+6)}$.
\end{proposition}
%\begin{proof}
%Lemma~\ref{lem:nosol} shows that there is no feasible signature $f$ with.
%\[
%\frac{k}{8k^2+10k+3}f_1 + \frac{1}{2(8k^2+10k+3)}f_2 \geq \frac{8k^2+6k+1}{8k^2+10k+3}.
%\]
%Therefore the solution in Lemma~\ref{lem:rellp} is feasible in~(\ref{lp}) and is then also optimal.
%\end{proof}

%The expressions for $s(K_2,K_{\frac{6k+5}{3k+1}})$ and $s(K_2,K_{\frac{8k+6}{4k+1}})$ have some similarities, but it is too early to draw any conclusions on how $s(K_2,K_{2+n/k})$ looks like when $n$ grows. For larger $n$ the number of orbits increases, and it becomes harder to determine if a set of signatures are complete or not.

The expressions for $s$ in Proposition~\ref{prop:3k+1}~and~\ref{prop:4k+1}
have some interesting similarities, but for $n \geq 5$ it becomes harder to
pick out a suitable set of solutions which guarantee that the relaxation has
the same optimum as (\ref{lp}) itself.
Using computer calculations, we have however determined the first
two values ($k = 1, 2$) for the case $t = 2+5/(5k+1)$
and the first value ($k = 1$) for the case $t = 2+6/(6k+1)$.
\begin{equation}
  s(K_2, K_{17/6})  = 322/425 \qquad 
  s(K_2, K_{27/11}) = 5/6 \qquad 
  s(K_2, K_{20/7})  = 67/89
\end{equation}

\subsection{Maps to Odd Cycles} \label{ssec:odd}

It was seen in Corollary~\ref{cor:intervals} that $s(K_r, K_{t})$ is constant
on the region $(r,t) \in [2,2+1/k) \times I_k$.
In this section, we will study what happens when $t$ remains in $I_k$, but
$r$ is set to $2+1/k$.
A first observation is that the absolute jump of the function
$s(K_r, K_{t})$ when $r$ goes from being less than $2+1/k$
to $r = 2+1/k$  must be largest for $t = 2+2/(2k-1)$.
Let $V(K_{2+2/(2k-1)}) = \{v_0, \ldots, v_{4k-1}\}$ and
$V(K_{2+1/k}) = \{w_0, \ldots, w_{2k}\}$.
The map $f(v_i) = w_i$ with the indices of $w$ taken modulo $2k+1$ has the
signature $f = (4k-1,2k)$.
Since the subgraph induced by the orbit $A_1$ is isomorphic to $C_{4k}$,
any map to an odd cycle must exclude at least one edge from $A_1$.
It follows that $f$ alone determines $s$, 
and we can solve (\ref{lp}) to obtain $s(K_{2+1/k}, K_{2+2/(2k-1)}) = (4k-1)/4k$.
Thus, for $r < 2+1/k$, we have
\begin{equation}
\label{eq:jump}
  s(K_{2+1/k}, K_{2+2/(2k-1)}) - s(K_r, K_{2+2/(2k-1)}) = (2k-1)/4k(2k+1)
\end{equation}
%We will now instead switch focus and examine how $s(C_{2k+1},K_r)$, with $2+1/k < r \leq 2+2/(2k-1)$, looks like. We have already seen in Theorem~\ref{cor:intervals} that for $a < 2+1/k$, $s(K_a,K_r)=s(K_2,K_r)$, and remember that $C_{2k+1}$ is isomorphic to $K_{2+1/k}$. It is therefore of interest to see how much of a jump the function makes at $a = 2+1/k$.
Smaller $t \in I_k$ can be expressed as $t = 2 + 1/(k-x)$, 
where $0 \leq x < 1/2$.
We will write $x = m/n$ for positive integers $m$ and $n$ which 
implies the form $t = 2 + n/(kn-m)$, with $m < n/2$.
For $m = 1$, it turns out to be sufficient to keep two inequalities 
from (\ref{lp}) to get an optimal value of $s$. 
From this we get the following result:

%\begin{exmp}
%\label{ex:solalpha}
%For $k=3$, $n=5$ and $m=3$ we get $K_{29/12}$ and the solution $f$ from Lemma~\ref{lem:solalpha} is constructed as follows:
%\[
%\begin{array}{c c c c c c c}
%{\bf f^{-1}(w_0)} & {\bf f^{-1}(w_1)} & {\bf f^{-1}(w_2)} & {\bf f^{-1}(w_3)} & {\bf f^{-1}(w_4)} & {\bf f^{-1}(w_5)} & {\bf f^{-1}(w_6)} \\ \hline
%v_0 & v_{17} & v_5 & v_{21} & v_9 & v_{25} & v_{13} \\
%v_1 & v_{18} & v_6 & v_{22} & v_{10} & v_{26} & v_{14} \\
%v_2 & v_{19} & v_7 & v_{23} & v_{11} & v_{27} & v_{15} \\
%v_3 & v_{20} & v_8 & v_{24} & v_{12} & v_{28} & v_{16} \\
%v_4 & & & & & &
%\end{array}
%\]
%We see that $v_0 v_{12}$, $v_{21} v_{4}$ and $v_{25} v_{8}$ are not mapped to $C_7$ in this solution, but all other edges are.
%\end{exmp}

%\begin{figure}[h]
%\centering
%\begin{tikzpicture}
%\SetVertexMath
%\grEmptyCycle[prefix=v,RA=5]{14}
%\tikzstyle{EdgeStyle}=[solid]
%\EdgeInGraphMod{v}{14}{5}
%\tikzstyle{EdgeStyle}=[dashed]
%\EdgeInGraphMod{v}{14}{6}
%\tikzstyle{EdgeStyle}=[dotted]
%\EdgeInGraphMod{v}{14}{7}
%\end{tikzpicture}
%\end{figure}

\begin{proposition} \label{prop:m1}
Let $k,n \geq 2$ be integers. Then, $s(C_{2k+1},K_{\frac{2(kn-1)+n}{kn-1}})=\frac{(2(kn-1)+n)(4k-1)}{(2(kn-1)+n)(4k-1)+4k-2}$.
\end{proposition}
There is still a non-zero jump of $s(K_r, K_{t})$ when we move from 
$K_r < 2+1/k$ to $K_r = 2+1/k$, but it is obviously smaller than that
of (\ref{eq:jump}) and tends to 0 as $n$ increases.
For $m = 2$, we have $2(kn-m)+n$ and $kn-m$ relatively prime only when $n$
is odd.
In this case, it turns out that we need to include an increasing number of
inequalities to obtain a good relaxation.
Furthermore, we are not able to ensure that the obtained value is the
optimum of the original (\ref{lp}).
We will therefore have to settle for a lower bound for $s$.
Explicit calculations have shown that, for small values of $k$ and $n$,
equality holds in Proposition~\ref{th:qisfunny}.
We conjecture this to be true in general.

\begin{proposition} \label{th:qisfunny}
Let $k \geq 2$ be an integer and $n \geq 3$ be an odd integer. Then,
\begin{equation}
s(C_{2k+1},K_{\frac{2(kn-2)+n}{kn-2}}) \geq \frac{(2(kn-2)+n)(\xi_n(4k-1)+(2k-1))}{(2(kn-2)+n)(\xi_n(4k-1)+(2k-1))+(4k-2)(1-\xi_n)},
\end{equation}
where
%\begin{equation*}
$\xi_n = 
\left(\alpha_1^{(n-1)/2} + \alpha_2^{(n-1)/2}\right)/4,$
%\left(1+\sqrt{\frac{2k+1}{4k-2}}\right)^{\frac{n-1}{2}}+
%\left(1-\sqrt{\frac{2k+1}{4k-2}}\right)^{\frac{n-1}{2}}
%\end{equation*}
and $\alpha_1, \alpha_2$ are the reciprocals of the roots of
$\frac{2k-3}{4k-2} z^2 - 2z + 1$.
\end{proposition}

\subsection{Extending the Results} \label{sec:apply}

We will now take a look at one possible way of extending the results in
the previous sections. To do this, we need to find graphs or classes of
graphs we can homomorphically sandwich between graphs with known
$s$ value. Clearly, $K_2$ has a homomorphism to all non-empty
graphs, and that if a graph $G$ has circular chromatic number $\chi_c(G)
\leq r$ it has a homomorphism to $K_r$. These facts, together with
Lemma~\ref{lem:sandwich}, combine into the following easily proved
lemma:

\begin{lemma}
\label{lem:circlesandwich}
Let $G$ be a non-empty graph with $\chi_c(G) \leq r$. Then, $s(K_2,G)
\geq s(K_2,K_r)$.
If, additionally, $G$ has odd girth no greater than $2k+1$, 
then $s(C_{2k+1},G) \geq s(C_{2k+1},K_r)$.
\end{lemma}
% \begin{proof}
% Since $\chi_c(G) \leq r$ means there exist one $t \leq r$ such that
%$G \rightarrow K_{t}$, and since $K_t \rightarrow K_r$ we have $G
%\rightarrow K_r$ and $K_2$ has a homomorphism to every graph that
%contains at least one edge so $K_2 \rightarrow G \rightarrow K_r$ and we
%can apply Lemma~\ref{lem:sandwich}. We also have that $C_{2k+1}$ has a
%homomorphism to each graph which contains an odd cycle with length at
%most $2k+1$. It is obvious that $C_{2k+1}$ has an homomorphism into a
%graph containing a cycle of length exactly $2k+1$. But we also know that
%$C_{2k+1} \rightarrow C_{2m+1}$ if $m \leq k$ so if $G$ contains an odd
%cycle of length at most $2k+1$ then we have $C_{2k+1} \rightarrow G
%\rightarrow K_r$. 
% \end{proof}

We can now make use of known results about bounds on the circular
chromatic number for certain classes of graphs. Much of the extensive
study conducted in this direction was instigated by the restriction of a
conjecture by Jaeger~\cite{jaeger:88} to planar graphs, which is
equivalent to the claim that every planar graph of girth at least $4k$
has a circular chromatic number at most $2 + 1/k$, for $k \geq 2$. The
case $k=1$ is Gr\"{o}tzsch's theorem; that every triangle-free planar
graph is 3-colourable. Currently, the best proven girth for when the
circular chromatic number of a planar graph is guaranteed to be at most
$2+1/k$ is $\frac{20k-2}{3}$ and due to Borodin et
al.~\cite{Borodin:etal:jctb2004}. This result was used by F\"{a}rnqvist
et al.\ to achieve the bound $s(K_2,G) \leq \frac{4k}{4k+1}$ for planar
graphs $G$ of girth at least $(40k-2)/3$. Here, we significantly improve
this bound by considering $K_4$-minor-free graphs, for which Pan and
Zhu~\cite{pan:zhu:02} have shown how their circular chromatic number is
upper-bounded by their odd girth.
\begin{proposition} \label{thmI}
Let $G$ be a $K_4$-minor-free graph, and $ k \geq 1$ an integer. If $G$
has an odd girth of at least $6k-1$, then $s(K_2,G) \leq
\frac{4k}{4k+1}$. If $G$ has an odd girth of at least $6k+3$, then
$s(K_2,G) \leq \frac{4k+2}{4k+3}$.
\end{proposition}
% \begin{proof}
% We use Theorem~\ref{th:zhu}, combined with Theorem~\ref{cor:intervals}
%to get values on $s(K_2,G)$. We get that when the odd girth is at least
%$6k-1$ then $s(K_2,G) \leq \frac{4k}{4k+1}$ and when the odd girth is at
%least $6k+3$ then $s(K_2,G) \leq \frac{4k+2}{4k+3}$, since $K_2
%\rightarrow G$ we have $d(K_2,G) = 1 - s(K_2,G)$.
% \end{proof}

Of course, it is a big limitation to only consider $K_4$-minor-free
graphs. Almost all work on the circular chromatic number for planar
graphs have focused on finding limits when $\chi_c(G) \leq 2 + 1/k$,
that is, when there exists a homomorphism to the odd cycle $C_{2k+1}$.
However, Corollary~\ref{cor:intervals} implies that for two graphs
$G$ and $H$, if $\chi_c(G) = 2 + 1/k$ and $\chi_c(H) = 2+ 2/(2k-1)$ then
$s(K_2,G)=s(K_2,H)$, so for our purposes it would be interesting to have
more results when $\chi_c(G) \leq 2+ 2/(2k-1)$. 
% \begin{corollary}\label{corI}
% Let $G$ be a $K_4$ minor-free graph with odd girth $5$. Then $d(C_5,G)
% \leq 3/31$.
% \end{corollary}
% % \begin{proof}
% % Since $G$ contains a cycle of length $5$, it is clear that $C_5
% \rightarrow G$, together with Lemma~\ref{lem:circlesandwich} and
% Theorem~\ref{cor:intervals} the results follow.
% % \end{proof}
For general graphs, we can use results from Raspaud and
Roussel~\cite{raspaud:rousell:07} relating the circular chromatic number
of graphs to their maximum average degree. Specifically, they show that
for a general graph $G$ of girth at least 12, 11, or 10, its circular
chromatic number is bounded from above by $8/3$, $11/4$, and $14/5$,
respectively, which translates into corresponding upper bounds $4/5$,
$17/22$, and $16/21$ on $s(K_2,G)$ (using Propositions \ref{prop:4k+4},
\ref{prop:3k+1}, \ref{prop:4k+1} and Lemma~\ref{lem:circlesandwich}).

\section{Fractional Covering by $H$-cuts} \label{sec:cut}

In the following, we slightly generalise the work of 
\v{S}\'{a}mal~\cite{samal:05,samal:06} on fractional
covering by cuts to obtain a complete correspondence between $s(H,G)$ 
and a family of `chromatic numbers' $\chi_H(G)$ which generalise 
\v{S}\'{a}mal's cubical chromatic number $\chi_q(G)$. 
The latter corresponds to the case when $H = K_2$.
First, we recall the notion of a {\em fractional colouring} of a
(hyper-) graph.
Let $G$ be a (hyper-) graph with vertex set $V(G)$ and edge set 
$E(G) \subseteq \mathcal{P}(V(G)) \setminus \{ \emptyset \}$.
A subset $I$ of $V(G)$ is called independent in $G$ if no edge 
$e \in E(G)$ is a subset of $I$.
Let $\mathcal{I}$ denote the set of all independent sets of $G$
and for a vertex $v \in V(G)$, let $\mathcal{I}$ 
denote all independent sets which contain $v$.
Then, the fractional chromatic number $\chi_f(G)$ of $G$ is given by the
linear program:
\begin{equation}
\begin{array}{ll}
            \text{Minimise}   & \sum_{I \in \mathcal{I}} f(I) \\
            \text{subject to} & \sum_{I \in \mathcal{I}(v)} f(I) \geq 1 \qquad \text{for all $v \in V(G)$}, \\ % where $f : {\mathcal I} \rightarrow \mathbb{R}^+$.} \\
	    \text{where}      & f : {\mathcal I} \rightarrow \mathbb{R}^+.
\end{array}
\end{equation}
%The concept mimics the relation between fractional colourings and Kneser graphs and we end the section by briefly elaborating on this.

The definition of fractional covering by cuts mimics fractional colouring,
but replaces vertices with edges and independent sets with certain cut
sets of the edges.
Let $G$ and $H$ be undirected simple graphs and $f$ be an arbitrary
vertex map from $G$ to $H$.
The map $f$ induces a partial map 
% $f^{\#} : E(G) \rightarrow E(H)$
from $E(G)$ to $E(H)$ and we will call the preimage of this map an {\em $H$-cut} in $G$.
% and we will call the preimage of $E(H)$ under $f^{\#}$ an {\em $H$-cut} in $G$.
%  set of edges $(f^{\#})^{-1}(H)$
When $H$ is a complete graph $K_k$, this is precisely the notion of a
{\em $k$-cut}.
Let $\mathcal{C}$ denote the set of $H$-cuts in $G$ and for an edge 
$e \in E(G)$, let $\mathcal{C}(e)$ denote all $H$-cuts which
contain $e$.
The following definition is the generalisation of
{\em cut $n/k$-covers}~\cite{samal:06} to arbitrary $H$-cuts:
\begin{definition}
  An {\em $H$-cut $n/k$-cover} of $G$ is a collection $X_1, \ldots, X_N$ of 
  $H$-cuts in $G$ such that every edge of $G$ is in at least $k$ of them.
  The graph parameter $\chi_H$ is defined as:
  \begin{equation}
  \chi_H(G) = \inf 
  \{ \frac{n}{k} \,|\, \text{there exists an $H$-cut $n/k$-cover of $G$.} \}
  \end{equation}
\end{definition}

By reasoning analogous to that of \v{S}\'{a}mal~\cite{samal:06} Lemma~5.1.3, $\chi_H$ is also
given by the following linear program:
\begin{equation} \label{coverprimal}
\begin{array}{ll}
            \text{Minimise}   & \sum_{X \in \mathcal{C}} f(X) \\
            \text{subject to} & \sum_{X \in \mathcal{C}(e)} f(X) \geq 1 \qquad \text{for all $e \in E(G)$}, \\ % where $f : {\mathcal C} \rightarrow \mathbb{R}^+$.} \\
	    \text{where}      & f : {\mathcal C} \rightarrow \mathbb{R}^+.
\end{array}
\end{equation}

For $H = K_2$, an alternative definition of $\chi_H(G) = \chi_q(G)$ was
obtained in~\cite{samal:06} by taking the infimum 
(actually minimum due to the formulation in (\ref{coverprimal}))
over $n/k$ for $n$ and $k$ such that $G \rightarrow Q_{n/k}$.
Here, $Q_{n/k}$ is the graph on vertex set $\{0,1\}^n$ with
an edge $u v$ if $d_H(u,v) \geq k$, where $d_H$ denotes the
Hamming distance.
We generalise this family as well to produce a scale for each $\chi_H$.
Namely, let $H^n_k$ be the graph on vertex set $V(H)^n$ and an edge
between $(u_1, \ldots, u_n)$ and $(v_1, \ldots, v_n)$ when
$|\{ i \,|\, (u_i, v_i) \in E(H) \}| \geq k$.
A moments thought shows that we can express $\chi_H$ as:
\begin{equation}
\chi_H(G) = \inf \{ \frac{n}{k} \,|\, G \rightarrow H^n_{k} \}.
\end{equation}
\v{S}\'{a}mal also notes that $\chi_q(G)$ is given by the fractional chromatic number
of a certain hypergraph associated to $G$.
For the general case, let $G'$ be the hypergraph obtained from $G$ by taking
$V(G') = E(G)$ and letting $E(G')$ be the set of minimal subgraphs 
$S \subseteq G$ such that $S \not\rightarrow H$.
A short argument shows that indeed $\chi_f(G') = \chi_H(G)$.

Finally, we can work out the correspondence to $s(H,G)$.
Consider the dual program of (\ref{coverprimal}):
\begin{equation} \label{coverdual}
\begin{array}{ll}
            \text{Maximise}   & \sum_{e \in E(G)} g(e) \\
            \text{subject to} & \sum_{e \in X} g(e) \leq 1 \qquad \text{for all $H$-cuts $X \in \mathcal{C}$}, \\ % where $g : E(G) \rightarrow \mathbb{R}^+$.} \\
	    \text{where}      & g : E(G) \rightarrow \mathbb{R}^+.
\end{array}
\end{equation}
Let $s = \sum_{e \in E(G)} g(e)$ and make the substitution $g' = g/s$ in 
(\ref{coverdual}).
Comparing with (\ref{lp}), we have
\begin{equation}
  \chi_H(G) = 1/s(H,G).
\end{equation}

We now move on to address two conjectures by \v{S}\'{a}mal~\cite{samal:06} on the cubical
chromatic number $\chi_q = \chi_{K_2}$.
In Section~\ref{sec:neg} we discuss an upper bound on $s$ which relates to
the first conjecture, Conjecture~5.5.3~\cite{samal:06}.
This is the suspicion
that $\chi_q(G)$ can be determined by measuring the maximum cut over
all subgraphs of $G$.
We show that this is false by providing a counterexample from
Section~\ref{ssec:k2}.
We then consider Conjecture~5.4.2~\cite{samal:06}, 
concerning ``measuring the scale'', i.e.,
determining $\chi_q$ for the graphs $Q_{n/k}$ themselves.
We prove that this conjecture is true, 
and state it as Proposition~\ref{prop:samal} in Section~\ref{sec:pos}.

\subsection{An Upper Bound on $s$} \label{sec:neg}

In Section~\ref{sec:meas} we obtained lower bounds on $s$ by relaxing the
linear program (\ref{lp}).
In most cases, the corresponding solution was proven feasible for the
original (\ref{lp}), and hence optimal.
Now, we take a look at the only known source of upper bounds for $s$.

Let $G, H \in \mathcal{G}$, with $G \rightarrow H$ and take an arbitrary
$S$ such that $G \rightarrow S \rightarrow H$.
Then, applying Lemma~\ref{lem:sandwich} followed by Lemma~\ref{lem:auto}
gives
\begin{equation} \label{eq:ub}
  s(G,H) \leq s(G,S) = \inf_{w \in {\cal {\hat W}}(S)} mc_G(S,w) 
  \leq mc_{G}(S, 1/|E(S)|).
\end{equation}
When $G = K_2$ it follows that
\begin{equation} \label{eq:samub}
  s(K_2,H) \leq \min_{S \subseteq G} b(S),
\end{equation}
where $b(S)$ denotes the bipartite density of $S$.
\v{S}\'{a}mal~\cite{samal:06} conjectured that
this inequality, expressed on the form
$\chi_q(S) \geq 1/(\min_{S \subseteq G} b(S))$,
can be replaced by an equality.
We answer this in the negative, using $K_{11/4}$ as our counterexample.
Lemma~\ref{prop:3k+1} with $k = 1$ gives $s(K_2,K_{11/4}) = 17/22$.
If $s(K_2,K_{11/4}) = b(S)$ for some $S \subseteq K_{11/4}$ it means
that $S$ must have at least $22$ edges. 
Since $K_{11/4}$ has exactly $22$ edges, then $S = K_{11/4}$.
However, a cut in a cycle must contain an even number of edges.
Since the edges of $K_{11/4}$ can be partitioned into two cycles,
we have that the maximum cut in $K_{11/4}$ must be of even size,
hence $|E(K_{11/4})| \cdot b(K_{11/4}) \neq 17$.
This is a contradiction.

\subsection{Confirmation of a Scale} \label{sec:pos}

As a part of his investigation of $\chi_q$, \v{S}\'{a}mal~\cite{samal:06}
set out to determine the value of $\chi_q(Q_{n/k})$.
We complete the proof of his Conjecture~5.4.2~\cite{samal:06}
to obtain the following result.
\begin{proposition} \label{prop:samal}
  Let $k, n$ be integers such that $k \leq n < 2k$.
  Then, 
  %\[
  $\chi_q(Q_{n/k}) = 
  %\begin{cases} 
  n/k$ if $k$ is even and % & \text{if $k$ is even and} \\
    $(n+1)/(k+1)$ if $k$ is odd. % & \text{if $k$ is odd.} \\
  %\end{cases}
  %\]
\end{proposition}
\v{S}\'{a}mal provides the upper bound and an approach to the lower bound 
using the largest eigenvalue of the Laplacian of a subgraph of $Q_{n/k}$.
The computation of this eigenvalue
boils down to an inequality (Conjecture~5.4.6~\cite{samal:06})
involving some binomial coefficients.
We first introduce the necessary notation and then prove the remaining
inequality in Lemma~\ref{lem:ineq}, whose second part, for odd $k$,
corresponds to one of the formulations of the conjecture.
Proposition~\ref{prop:samal} then follows from
Theorem~5.4.7~\cite{samal:06} conditioned on the result of this lemma.

  Let $k, n$ be positive integers such that $k \leq n$, and let
  $x$ be an integer such that $1 \leq x \leq n$.
  For $k \leq n < 2k$, let $S_o(n,k,x)$ denote the set of all $k$-subsets of $\{1, \ldots, n\}$ that have an odd number of elements in common with the set $\{n-x+1, \ldots, n\}$.
  Define $S_e(n,k,x)$ analogously as the $k$-subsets with an even number of common elements.
  Let $N_o(n,k,x) = |S_o(n,k,x)|$ and $N_e(n,k,x) = |S_e(n,k,x)|$. Then, % = \sum_{\mathrm{odd} \, t} \binom{x}{t} \binom{n-x}{k-t}$ and $N_e(n,k,x) = |S_e(n,k,x)| = \sum_{\mathrm{even} \, t} \binom{x}{t} \binom{n-x}{k-t}$.
  \begin{equation}
  N_o(n,k,x) = \sum_{odd \, t} \binom{x}{t} \binom{n-x}{k-t}, \quad
  N_e(n,k,x) = \sum_{even \, t} \binom{x}{t} \binom{n-x}{k-t}.
  \end{equation}
  %We are going to find an upper bound, for $k \leq n < 2k$, on the number of $k$-subsets of $\{1, \ldots, n\}$ which have an odd (even) number of elements in common with the set $\{n-x+1, \ldots, n\}$.
  %Define the sets
  %\[
  %S_o(n,k,x) = \{ \sigma \in \binom{\{1,\ldots,n\}}{k} \,|\, |\sigma \cap \{n-x+1, \ldots, n\}| \text{ is odd} \},
  %\]
  %and
  %\[
  %S_e(n,k,x) = \{ \sigma \in \binom{\{1,\ldots,n\}}{k} \,|\, |\sigma \cap \{n-x+1, \ldots, n\}| \text{ is even} \},
  %\]
  %and let them be counted by
  %\[
  %N_o(n,k,x) := |S_o(n,k,x)| = \sum_{odd \, t} \binom{x}{t} \binom{n-x}{k-t},
  %\]
  %and
  %\[
  %N_e(n,k,x) := |S_e(n,k,x)| = \sum_{even \, t} \binom{x}{t} \binom{n-x}{k-t}.
  %\]
  %These numbers naturally count the following sets:
  %\[
  %A_o(n,k,x) = \{ S \in \binom{\{1,\ldots,n\}}{k} \,|\, S \cap
  %\{n-x+1, \ldots, n\} \text{ is odd,} \}
  %\]
  %and
  %\[
  %A_e(n,k,x) = \{ S \in \binom{\{1,\ldots,n\}}{k} \,|\, S \cap
  %\{n-x+1, \ldots, n\} \text{ is even.} \}
  %\]

  When $x$ is odd, the function $f : S_o(2k,k,x) \rightarrow S_e(2k,k,x)$
  given by the complement $f(\sigma) = \{1, \ldots, n\} \setminus \sigma$
  is a bijection.
  Since $N_o(n,k,x)+N_e(n,k,x) = \binom{n}{k}$, we have
  \begin{equation} \label{eqn:1}
    N_o(2k,k,x) = N_e(2k,k,x) = \frac{1}{2} \binom{2k}{k}.
  \end{equation}

  \begin{lemma} \label{lem:help}
    Let $1 \leq x < n = 2k-1$ with $x$ odd. Then,
    $N_e(n,k,x) = N_e(n,k,x+1)$ and $N_o(n,k,x) = N_o(n,k,x+1)$.
  \end{lemma}

  \begin{proof}
    First, partition $S_e(n,k,x)$ into $A_1 = \{ \sigma \in S_e(n,k,x) \,|\, n-x \not\in \sigma \}$ and $A_2 = S_e(n,k,x) \setminus A_1$.
    Similarly, partition $S_e(n,k,x+1)$ into $B_1 = \{ \sigma \in S_e(n,k,x+1) \,|\, n-x \not\in \sigma \}$ and $B_2 = S_e(n,k,x+1) \setminus B_1$.
    Note that $A_1 = B_1$.
    We argue that $|A_2| = |B_2|$.
    To prove this, define the function $f : \mathcal{P}(\{1,\ldots,n\}) \rightarrow \mathcal{P}(\{1,\ldots,n-1\})$ by
    %\[
    $f(\sigma) = (\sigma \cap \{1, \ldots, n-x-1\}) \cup \{ s-1 \,|\; s \in \sigma, s > n-x \},$
    %\]
    i.e., $f$ acts on $\sigma$ by ignoring the element $n-x$ and renumbering
    subsequent elements so that the image is a subset of $\{1, \ldots, n-1\}$.
    Note that $f(A_2) = S_e(2k-2,k-1,x)$ and $f(B_2) = S_o(2k-2,k-1,x)$.
    Since $x$ is odd, it follows from (\ref{eqn:1}) that
    $|f(A_2)| = |f(B_2)|$.
    The first part of the lemma now follows from the injectivity of
    the restrictions $f|_{A_2}$ and $f|_{B_2}$.
    The second equality is proved similarly.
%    \qed    
  \end{proof}

  \begin{lemma} \label{lem:ineq}
    Choose $k, n$ and $x$ so that $k \leq n < 2k$ and $1 \leq x \leq n$.
    For odd $k$,
    \begin{equation}
    N_e(n,k,x) \leq \binom{n-1}{k-1} \quad \text{and for even $k$,} \quad N_o(n,k,x) \leq \binom{n-1}{k-1}.
    \end{equation}
    %and for odd $k$,
    %\[
    %N_o(n,k,x) \leq \binom{n-1}{k-1}.
    %\]
  \end{lemma}
  
  \begin{proof}
    We will proceed by induction over $n$ and $x$.
    The base cases are given by $x = 1$, $x = n$, and $n = k$.
    For $x = 1$,
    %\[
    $N_o(n,k,x) = \binom{n-1}{k-1}$ and $N_e(n,k,x) = \binom{n-1}{k} \leq \binom{n-1}{k-1}$,
    %\]
    where the inequality holds for all $n < 2k$.
    For $x = n$ and odd $k$, we have $N_e(n,k,x) = 0$, and for
    even $k$, we have $N_o(n,k,x) = 0$.
    For $n = k$,
    %\[
    $N_e(n,k,x) = 1-N_o(n,k,x) = 1$ if $x$ is even and 0 otherwise.
    %\begin{cases} 
    %  1 & \text{if $x$ is even,} \\
    %  0 & \text{otherwise.} \\
    %\end{cases}
    %\]
    Let $1 < x < n$ and consider $N_e(n,k,x)$ for odd $k$ and $k < n < 2k-1$.
    Partition the sets $\sigma \in S_e(n,k,x)$ into those for which 
    $n \in \sigma$ on the one hand and those for which
    $n \not\in \sigma$ on the other hand.
    These parts contain $N_o(n-1,k-1,x-1)$ and $N_e(n-1,k,x-1)$ sets,
    respectively.
    Since $k-1$ is even, and since $k \leq n-1 < 2(k-1)$ when $k < n < 2k-1$, 
    it follows from the induction hypothesis that
    %\begin{multline*}
      $N_e(n,k,x) =
      N_o(n-1,k-1,x-1) + N_e(n-1,k,x-1) \leq
      \binom{n-2}{k-2} + \binom{n-2}{k-1} = \binom{n-1}{k-1}.$
    %\end{multline*}
    The case for $N_o(n,k,x)$ and even $k$ is treated identically.

    Finally, let $n = 2k-1$.
    If $x$ is odd, then Lemma~\ref{lem:help} is applicable, 
    so we can assume that $x$ is even.
    %If $x = n$, then $N_e(n,k,x) = 0$ since $k$ is even.
    %Otherwise, if $x$ is odd, then Lemma~\ref{lem:help} is applicable, 
    %so we can assume that $x$ is even.
    Now, as before
    %\begin{multline*}
    $N_e(2k-1,k,x) =
    N_o(2k-2,k-1,x-1) + N_e(2k-2,k,x-1) \leq  
    \frac{1}{2} \binom{2k-2}{k-1} + \binom{2k-3}{k-1} = \binom{n-1}{k-1},$
    %\end{multline*}
    where the first term is evaluated using (\ref{eqn:1}).
    The same inequality can be shown for $N_o(2k-1,k,x)$
    and even $k$,
    which completes the proof.
%    \qed
  \end{proof}

\section{Conclusions and Open Problems}
\label{sec:open}

We have seen that for all integers $k \geq 2$, $s(K_2,K_t)$ is constant on $I_k$.
It follows that our sandwich approach using Lemma~\ref{lem:sandwich} with $M = K_2$
and $N = K_r$ can not distinguish between the class
of graphs with circular chromatic number $2+1/k$ and the (larger) class with
circular chromatic number $2+2/(2k-1)$.
As previously noted, Jaeger's conjecture and subsequent research
has provided partial information on the members of the former class.
We remark that Jaeger's conjecture implies a weaker statement in our
setting. Namely, if $G$ is a planar graph with girth greater than $4k$,
then $G \rightarrow C_k$ implies $s(K_2, G) \geq s(K_2,C_k) =
2k/(2k+1)$. Deciding this to be true would certainly provide support for
the original conjecture, and would be an interesting result in its
own right.
Our starting observation shows that the slightly weaker condition
$G \rightarrow K_{2+2/(2k-1)}$ implies the same result.

When it comes to completely understanding how $s$ behaves on circular complete graphs, even
restricted to those between $K_2$ and $K_3$, there is still work to be done.
%The regions found in this paper seem to ... the edge-transitive graphs $K_r$, i.e.\ , the cycles and the complete graphs.
For edge-transitive graphs $K_t$, in our case the cycles and the complete graphs, 
it is not surprising
that the expression $s(K_r, K_t)$ assumes a finite number of values seen as a function of $r$.
Indeed, Lemma~\ref{lem:auto} says that $s(K_r, K_t) = mc_{K_r}(K_t, 1/|E(K_t)|)$ which
leaves at most $|E(K_t)|$ values for $s$.
This produces a number of constant intervals which are partly 
responsible for the constant regions of Corollary~\ref{cor:intervals} and the discussion
following it.
More surprising are the constant intervals that arise from
$s(K_r,K_{2+2/(2k-1)})$.
They give some hope that the behaviour of $s$ is possible to characterise more generally.
One direction could be to identify additional constant regions,
perhaps showing that they completely tile the entire space?

%If Conjecture~\ref{conj:cyc-2} is true, then it seems likely that there is no simple way to 
%calculate the distance between two arbitrary circular complete graphs between $K_2$ and $K_3$. 
%The possibility to completely determine the distance from $K_2$ to all other circular complete 
%graphs smaller than $K_3$ seems like a more manageable task, even though 
%the expressions will become more and more complicated as the number of orbits grows. 
%Another thing that would be interesting to look into is what properties 
%the circular complete graphs between $K_{\frac{2k+1}{k}}$ and $K_{\frac{4k}{2k-1}}$ have that 
%makes their distance to other circular complete graphs constant. Also does this somehow translate 
%into the circular chromatic number of other graphs. Do graphs that have a circular chromatic number 
%between $2 + 1/k$ and $2+2/(2k-1)$ share some common characteristic that graphs whose circular 
%chromatic numbers falls outside of this interval lack?

In Section~\ref{sec:cut} we generalised the notion of covering by cuts
due to \v{S}\'{a}mal.
By doing this, we have found a different interpretation of the $s$-numbers
as an entire family of `chromatic numbers'.
It is our belief that these alternate viewpoints can benefit from each other.
The refuted conjecture in Section~\ref{sec:neg} is an immediate example of
this.
On the other hand, it would be interesting to determine
when the generalised upper bound in (\ref{eq:ub}) is tight.
For $H = K_2$, the proof of Proposition~\ref{prop:samal} is precisely such a result
for the graphs $Q_{n/k}$,
which is evident from studying the proof of Theorem 5.4.7~\cite{samal:06}.
Following this, a natural step would be to calculate $\chi_H(H_k^n)$ for
more general graphs $H$, starting with $H = K_3$.

It is fairly obvious that $\MCol{H}$ is a special case of the 
{\em maximum constraint satisfaction} ({\sc Max CSP}) problem;
in this problem, one is given a finite collection of constraints on overlapping
sets of variables, and the goal is to assign values from a given domain to the
variables so as to maximise the number of satisfied constraints.
By letting $\Gamma$ be a finite set of relations, we can
parameterise {\sc Max CSP} with $\Gamma$ ({\sc Max CSP}$(\Gamma)$) so that
the only allowed constraints are those constructed from the relations in $\Gamma$.
By viewing a graph $H$ as a binary relation, the problems {\sc Max CSP}$(\{H\})$
and $\MCol{H}$ are virtually identical. 
Raghavendra~\cite{raghavendra:08} has presented
an algorithm for {\sc Max CSP}$(\Gamma)$ based on semi-definite programming.
Under the so-called {\em unique games conjecture}, this algorithm
optimally approximates {\sc Max CSP}$(\Gamma)$ in polynomial-time, i.e. no
other polynomial-time algorithm can approximate the problem substantially better.
However, it is notoriously difficult to find out exactly how well the
algorithm approximates {\sc Max CSP}$(\Gamma)$ for a given $\Gamma$.
It seems plausible that
the function $s$ can be extended into a function $s'$ from pairs of sets
of relations to ${\mathbb Q}^+$, and that $s'$ can be used for studying
the approximability of {\sc Max CSP} by extending the approach in
F\"arnqvist~et al.~\cite{farnqvist:etal:09}. This would constitute a novel method for
studying the approximability of {\sc Max CSP} --- a method that, hopefully, may
cast some new light on the performance of Raghavendra's algorithm.

\bibliography{wg}

\newpage

\appendix
\begin{center}
  {\bf APPENDIX}
\end{center}

% In this section we will look at the family of graphs that is called circular complete graphs. We will see that these are a refinement of the normal complete graphs. We will also look at the closely related topic of circular colourings, like normal colourings can be seen as homomorphisms to the complete graphs, circular colourings can be seen as homomorphisms to the circular complete graphs.

\noindent
Let  $0 < q \leq p$ be positive integers. We often assign names to the vertices, so that $V(K_{p/q}) = \{v_0,v_1,\ldots,v_{p-1}\}$. Then, we have $E(K_{p/q}) = \{v_i v_j \; | \; q \leq |i-j| \leq p-q\}$. Note that $K_{p/q}$ does not have any edges unless $p \geq 2q$, since the circular distance between two vertices is as most $p/2$. 
%It is also easy to see that $K_{p/1}$ is isomorphic to $K_p$, as any vertex has edges to all other vertices. Finally we can conclude that $K_{(2k+1)/k}$ is isomorphic to the odd cycle $C_{2k+1}$. We realise that each vertex has exactly 2 edges, and the edges form a cycle since $k$ and $2k+1$ are relatively prime. 
For a fixed $p$, let $\delta(v_i,v_j) = j - i$ (mod $p$). $\delta(v_i,v_j)$ is then the directed circular distance (in positive direction) between $v_i$ and $v_j$. Furthermore let $\bar{\delta}(v_i,v_j) = \min{\{\delta(v_i,v_j), \delta(v_j,v_i)\}}$. This is then the undirected circular distance. We do index arithmetics for circular complete graphs modulo $p$, e.g.
$v_{-1}=v_{p-1}$. Even though $K_{2k+1/k}$ is isomorphic to $C_{2k+1}$, we distinguish them by letting $v_i v_j$ be an edge in $C_{2k+1}$ if $\bar{\delta}_{2k+1}(v_i,v_j)=1$, while $v_i v_j$ is an edge in $K_{2k+1/k}$ if $\bar{\delta}_{2k+1}(v_i,v_j)=k$.

Let $M$ and $N$ be graphs and
let $F$ be a set of signatures to $(N,\omega)$ of $\MCol{M}$.
If $F' \subseteq F$ is a subset for which the relaxation of (\ref{lp}) 
has the same optimal solution as the original program,
we will call $F'$ a \emph{complete} set of signatures with respect to
$(N,\omega)$ of $\MCol{M}$.

\section{Proofs of Results from Section~\ref{ssec:k2}}

\subsection*{Proposition~\ref{prop:3k+1}}
\begin{proof}
Let $V(K_{\frac{6k+5}{3k+1}}) = \{v_0,v_1,\ldots,v_{6k+4}\}$ and $V(K_2) = \{w_0, w_1\}$. 
%A complete set of signatures is $\{f,h\}$ with $f=(6k+2,6k+4)$ and $h=(6k+4,2k+2)$.
Let $f$ be the solution with $f(v_i) = w_0$ if $0 \leq i < 3k+3$ and $f(v_i) = w_1$
if $3k+3 \leq i < 6k+5$.
From $A_1$ only the edges $v_0 v_{3k+1}$, $v_1 v_{3k+2}$ and $v_{3k+3} v_{6k+4}$ are mapped to a single vertex in $K_2$, so $f_1=6k+2$. From $A_2$ only the edge $v_0 v_{3k+2}$ is mapped to a single vertex in $K_2$, so $f_2=6k+4$.
Thus, $f$ has the signature $f = (6k+2, 6k+4)$.

%Define $h$ by $h(v_0)=h(v_3)=\cdots=h(v_{6k+3})=h(v_1)=\cdots=h(v_{3k+1})=w_0$ and $h(v_{3k+4})=h(v_{3k+7})=\cdots=h(v_{6k+4})=h(v_2)=\cdots=h(v_{6k+2})=w_1$.
%From $A_1$ only the endpoints from $v_0 v_{3k+1}$ are mapped to the same vertex in $K_2$. From $A_2$ the endpoints of $v_0 v_{3k+2},v_3 v_{3k+5},\ldots,v_{3k} v_{6k+2}$ as well as $v_{3k+4} v_1,v_{3k+7} v_4,\ldots,v_{6k+4} v_{3k+1}$ are mapped to the same vertex in $K_2$.
%Thus, $h$ has the signature $h = (6k+4, 2k+2)$.

Note that since $6k+5$ and $3k+2$ are relatively prime, the edges of $A_1$, 
as well as $A_2$, form cylces of length $6k+5$.
Therefore, any solution which maps more than $6k+2$ edges from $A_1$ to $K_2$ 
must map exactly $6k+4$.
Let $g$ be such a solution. We will show that $g_2 = 2k+2$.
%Note that only one edge in $A_1$ will not be mapped to $K_2$ by $g$.
We may assume that $v_{3k+1} v_0$ is the edge in $A_1$ which is not
mapped to $K_2$ by $g$.
Note that, if $i \neq 3k+1, 6k+2$, then $v_i v_{i+3k+4}$ and $v_{i+3k+4} v_{i+3}$
are both mapped to $K_2$ by $g$ which implies that $g(v_i) = g(v_{i+3})$.
Now, let $v_l v_{l+3(k+1)}$ be an edge in $A_2$ and let
$S = \{l, l+3, \ldots, l+3k\}$.
Then, this edge is mapped to $K_2$ by $g$, i.e., $g(v_l) \neq g(v_{l+3(k+1)})$ 
if and only if
$\{3k+1, 6k+2\} \cap S \neq \emptyset$.
Since $v_{3k+1}$ and $v_{6k+2}$ are adjacent in $A_1$, they can not
both be in $S$.
Therefore, there are $2 \cdot |S| = 2(k+1)$ edges that are mapped
to $K_2$ by $g$, so $g = (6k+4, 2k+2)$.
We conclude that solving (\ref{lp}) with the inequalities obtained from
$f$ and $g$ yields the correct value of $s$.
%The signatures $f = (6k+2,6k+4)$ and $h = (6k+4,2k+2)$ forms a complete set with respect to $(K_{\frac{6k+5}{3k+1}},\omega$) of {\sc MAX $K_2$-COL}. The solution to~(\ref{lp}) with $\omega_1,\omega_2,s$ as basic variables give reduced costs $\left(\frac{6k^2+7k+2}{6k^2+11k+5}, \frac{3k+2}{6k^2+11k+5}\right)$ and $s = \frac{6k^2+8k+3}{6k^2+11k+5}$.
\end{proof}

\subsection*{Proposition~\ref{prop:4k+1}}
\begin{proof}
%Let $k \geq 1$ be an integer. Then, there exists two signatures $f=(8k+4,4k+4,2k+1)$ and $h=(8k+2,8k+4,4k+3)$ to $(K_{\frac{8k+6}{4k+1}},\omega$) of {\sc MAX $K_2$-COL}.
Let $V(K_{\frac{8k+6}{4k+1}}) = \{v_0,v_1,\ldots,v_{8k+5}\}$ and $V(K_2) = \{w_0,w_1\}$. 
Define $f$ by $h(v_i) = w_0$ if $0 \leq i < 4k+3$ and $f(v_i) = w_1$ if $4k+3 \leq i < 8k+6$.
Here, the edges $v_0 v_{4k+1},v_1 v_{4k+2},v_{4k+3} v_{8k+4}$ and 
$v_{4k+4} v_{8k+5}$ in $A_1$ are mapped to a single vertex in $K_2$ by $f$
From $A_2$, $f$ maps edges $v_0 v_{4k+2}$ and $v_{4k+3} v_{8k+5}$ to a single 
vertex in $K_2$.
Finally, $f$ maps all edges in $A_3$ to the edge in $K_2$.
The signature of this solution is $f=(8k+2,8k+4,4k+3)$.

Let $g$ be defined by
\begin{multline*}
g(v_0)=g(v_4)=\cdots=g(v_{8k+4})=g(v_2)=\cdots \\
=g(v_{4k-2})=g(v_{8k+3})=g(v_1)=\cdots=g(v_{4k-3})=w_0
\end{multline*}
and 
\begin{multline*}
g(v_{4k+1})=g(v_{4k+5})=\cdots=g(v_{8k+5})=g(v_3)=\cdots \\
=g(v_{8k-1})=g(v_{4k+2})=g(v_{4k+6})=\cdots=g(v_{8k+2})=w_1.
\end{multline*}
From $A_1$ only the edges $v_{4k+1} v_{8k+2}$ and $v_{8k+3} v_{4k-2}$ are mapped to
a single vertex in $K_2$. From $A_2$ we partition the edges which are mapped to
the edge in $K_2$ by $g$ into four sets, with $k+1$ edges in each set. These are
\[
\{ v_0 v_{4k+2},v_4 v_{4k+6},\ldots,v_{4k} v_{8k+2} \},
\]
\[
\{ v_{4k+2} v_{8k+4}, v_{4k+6} v_2,\ldots,v_{8k+2} v_{4k-2} \},
\]
\[
\{ v_{4k+1} v_{8k+3},v_{4k+5} v_1,\ldots,v_{8k+1} v_{4k-3} \},
\]
\[
\{ v_{8k+3} v_{4k-1}, v_1 v_{4k+3},\ldots,v_{4k-3} v_{8k-1} \}.
\]
Finally, for $A_3$, $g$ maps the $k$ edges 
$v_0 v_{4k+3},v_4 v_{4k+7},\ldots,v_{4k-4} v_{8k-1}$ as well as the $k+1$ edges 
$v_{4k+1} v_{8k+4},v_{4k+5} v_2,\ldots,v_{8k+1} v_{4k-2}$ to the edge in $K_2$.
In summary, $g=(8k+4,4k+4,2k+1)$.

The relaxation of (\ref{lp}) corresponding to the two solutions $f$ and $g$ has
the following solution:
\[
s=\frac{8k^2+6k+2}{8k^2+10k+3}, \quad \omega_1=\frac{k}{8k^2+10k+3}, \quad \omega_2=\frac{1}{2(8k^2+10k+3)}, \quad \omega_3=0.
\]
We will now show that $s, \omega_1, \omega_2$ and $\omega_3$ is feasible in
the original program.
We will show that for all solutions $h$, we must have $h_2 \leq 8k+4$.
We will also show that if $h$ is such that $h_1 = 8k+4$, then
$h_2 \leq 4k+4$.
Finally, we will show that if $h_1 = 8k+6$, then $h_2$ must be 0.
In the final case, we note that $\omega_1 \cdot h_1 + \omega_2 \cdot h_2 < s$.

The edges of $A_2$ connects vertices at a distance of $4k+2$.
Since we have a common factor $2$ in $4k+2$ and $8k+6$,
the edges of $A_2$ consists of two odd cycles, each of length $4k+3$.
Since a cut of a cycle must include an even number of edges, 
we can then at most have a solution that maps $8k+4$ edges to $K_2$.

For the second case, 
note that $v_{i+4k+2}=v_{i+(2k+2)(4k+1)}$. 
This means that the shortest path between $v_i$ and $v_{i+4k+2}$ in $A_1$
is of length $2k+2$. The edge $v_i v_{i+4k+2}$ is mapped to $K_2$ 
if and only if at least one edge in each of the paths from $v_i$ to $v_{i+4k+2}$ in
 $A_1$ is not mapped to $K_2$, since they are both of even length.
If a solution $h$ has $h_1=8k+4$, only two edges from $A_1$ are
 not mapped to $K_2$.
Therefore no more than $4k+4$ paths of length $2k+2$ can include 
at least one of these two edges, hence $h_2 \leq 4k+4$.

Finally, if a solution $h$ includes an edge from $A_2$ it means that 
$h(v_i) \neq h(v_{i+4k+2})$ for some $i$.
But since both paths from $v_i$ to $v_{i+4k+2}$ in $A_1$ are of even length,
not all edges from $A_1$ can be mapped to $K_2$. 
So if $h_2 > 0$, then $h_1 < 8k+6$.

%We have shown that there is no feasible signature $h$ with
%\[
%\frac{k}{8k^2+10k+3}h_1 + \frac{1}{2(8k^2+10k+3)}h_2 \geq \frac{8k^2+6k+1}{8k^2+10k+3}.
%\]
%Therefore our solution is feasible in~(\ref{lp}) and, hence, also optimal.
\end{proof}

\section{Proof of Proposition~\ref{prop:m1}} % of Results from Section~\ref{ssec:odd}}
\label{app:propm1proof}

The proof of Proposition~\ref{prop:m1} follows from Lemma~\ref{lem:solalpha} and~\ref{lem:solbeta} introduced and proved in this section.

\subsection*{Proposition~\ref{prop:m1}}
\begin{proof}
  Let $p = 2(kn-1)+n$.
  From Lemma~\ref{lem:solalpha}, we get a solution $f$, with
  \begin{equation}
  f = (\alpha \cdot |A_1|, |A_2|, \ldots, |A_{\lceil \frac{n+1}{2} \rceil}|),
  \end{equation}
  where $\alpha = 1-1/p$.
  From Lemma~\ref{lem:solbeta}, we get another solution $f'$, with
  \begin{equation}
  f' = (|A_1|, \beta \cdot |A_2|, \ldots, \beta \cdot |A_{\lceil \frac{n+1}{2} \rceil}|),
  \end{equation}
  where $\beta = 1 - 2(2k-1)/p$.
  The last constraint in (\ref{lp}) can be written as
  \begin{equation}
    \label{eq:isol1}
    \sum_{i \neq 1} \omega_k \cdot |A_i| = 1 - \omega_1 \cdot |A_1|.
  \end{equation}
  We now insert (\ref{eq:isol1}) into the inequalities obtained from $f$ and
  $f'$ to get the following relaxation of (\ref{lp}):
  \begin{equation}
    \label{eq:prog}
    \begin{array}{l}
      \omega_1 \cdot |A_1| \cdot (\alpha-1) + 1 \leq s \\
      \omega_1 \cdot |A_1| \cdot (1-\beta) + \beta \leq s. \\
    \end{array}
  \end{equation}
  The solution to this is $\frac{1-\alpha\beta}{2-\alpha-\beta}$,
  which yields the $s$-value in the proposition.

  To show that this is optimal for the original program,
  let us consider the restriction of (\ref{lp}) in which we force
  $\omega_i = 0$ for $i = 3, \ldots, \lceil \frac{n+1}{2} \rceil$.
  Due to the second part of Lemma~\ref{lem:solbeta}, it suffices to
  keep the two inequalities from $f$ and $f'$ in the program.
  The equality constraint can now be written as
  \begin{equation}
    \label{eq:isol2}
    \omega_2 \cdot |A_2| = 1 - \omega_1 \cdot |A_1|.
  \end{equation}
  By inserting (\ref{eq:isol2}) into the two remaining
  inequalities we again obtain (\ref{eq:prog}).
  Thus, the solution to the relaxation gives the right value for $s$.
\end{proof}

\begin{lemma}
\label{lem:solalpha}
Let $k,n,m$ be integers with $k,n \geq 2$ and $1 \leq m \leq \min\{n/2$, $2k+1\}$. Then, there exist a solution $f$ to $(K_{\frac{2kn+n-2m}{kn-m}},\omega)$ of MAX $C_{2k+1}$-COL with signature $(|A_1|-m, |A_2|, \ldots,|A_{\lceil\frac{n+1}{2}\rceil}|)$. 
\end{lemma}

%\subsection*{Lemma~\ref{lem:solalpha}}
\begin{proof}
Let $V(K_{\frac{2kn+n-2m}{kn-m}}) = \{v_0,\ldots,v_{2kn+n-2m-1}\}$ and $V(C_{2k+1}) = \{w_0$, $\ldots$, $w_{2k}\}$. The construction of $f$ will depend on whether $m \leq k$ or $m > k$.
When $m \leq k$ we define $f$ as follows.

$f^{-1}(w_0) = \{v_0,v_1,\ldots,v_{n-1}\}$,

$f^{-1}(w_2) = \{v_{n},\ldots,v_{2n-1}\}$,

$\vdots$ 

$f^{-1}(w_{2k-2m}) = \{v_{(k-m)n},\ldots,v_{(k-m+1)n-1}\}$,

$f^{-1}(w_{2k-2m+2}) = \{v_{(k-m+1)n},\ldots,v_{(k-m+2)n-2}\}$,

$\vdots$

$f^{-1}(w_{2k}) = \{v_{kn-m+1},\ldots,v_{(k+1)n-m-1}\}$,

$f^{-1}(w_1) = \{v_{(k+1)n-m},\ldots\,v_{(k+2)n-m-1}\}$,

$\vdots$

$f^{-1}(w_{2k-2m-1}) = \{v_{(2k-2m-1)n-m},\ldots,v_{(2k-2m)n-m-1}\}$,

$f^{-1}(w_{2k-2m+1}) = \{v_{(2k-2m)n-m},\ldots,v_{(2k-2m+1)n-m-2}\}$,

$\vdots$

$f^{-1}(w_{2k-1}) =\{v_{2kn-2m+1},\ldots,v_{(2k+1)n-2m-1}\}$.

\noindent
Note, in particular, that
\[
|f^{-1}(w_j)| = \begin{cases}
  n & \text{for $0 \leq j \leq 2(k-m)$, and} \\
  n-1 & \text{for $2(k-m) < j \leq 2k-1$.} \\
\end{cases}
\]
When $m > k$, we define $f$ as follows:

$f^{-1}(w_0) = \{v_0,v_1,\ldots,v_{n-2}\}$, 

$f^{-1}(w_2) = \{v_{n-1},\ldots,v_{2n-3}\}$,

$\vdots$ 

$f^{-1}(w_{4k-2m}) = \{v_{(2k-m)(n-1)},\ldots,v_{(2k-m+1)(n-1)-1}\}$,

$f^{-1}(w_{4k-2m+2}) = \{v_{(2k-m+1)(n-1)},\ldots,v_{(2k-m+2)(n-1)-2}\}$,

$\vdots$

$f^{-1}(w_{2k}) = \{v_{k(n-1)-m+k+1},\ldots,v_{(k+1)(n-1)-m+k-1}\}$,

$f^{-1}(w_1) = \{v_{(k+1)(n-1)-m+k},\ldots\,v_{(k+2)(n-1)-m+k-1}$,

$\vdots$

$f^{-1}(w_{4k-2m+1}) = \{v_{(4k-2m)(n-1)-m+k},\ldots,v_{(4k-2m+1)(n-1)-m+k-1}\}$,

$f^{-1}(w_{4k-2m+3}) = \{v_{(4k-2m+1)(n-1)-m+k},\ldots,v_{(4k-2m+2)(n-1)-m+k-2}\}$,

$\vdots$

$f^{-1}(w_{2k-1}) =\{v_{2k(n-1)-2m+2k+2},\ldots,v_{(2k+1)(n-1)-2m+2k}\}$.

\noindent
In this case,
\[
|f^{-1}(w_j)| = \begin{cases}
  n-1 & \text{for $0 \leq j < 2(2k-m+1)$, and} \\
  n-2 & \text{for $2(2k-m+1) \leq j \leq 2k-1$.} \\
\end{cases}
\]

Now,
consider a vertex $v_i$ with $f(v_i) = w_j$. Take one edge $v_i v_l \in A_2 \cup \cdots \cup A_{\lceil\frac{n+1}{2}\rceil}$. Then,
\begin{equation} \label{eq:deltail}
kn-m+1 \leq \delta(v_i,v_l) \leq 2(kn-m)+n-(kn-m+1)=kn-m+n-1.
\end{equation}
Let $a=\min{\{h \;|\; f(v_h) = w_{j-1} \}}$. That is, $v_a$ is the vertex with lowest index which is mapped to $w_{j-1}$. Furthermore let $b=\max{\{h \;|\; f(v_h) = w_{j+1} \}}$. We then have
\[
f(\{v_a,v_{a+1},\ldots,v_{b-1},v_b\})=\{w_{j-1},w_{j+1}\}.
\]
We now want to show that $l \in \{a, \ldots, b\}$. It will then follow that $f(v_i) f(v_l) \in E(C_{2k+1})$, i.e., all edges outside of $A_1$ are mapped to an edge in $C_{2k+1}$.
To do this, we will show that $\delta(v_i,v_a) \leq \delta(v_i,v_l) \leq \delta(v_i,v_b)$.

First, we bound $\delta(v_i,v_a)$ from above by taking a walk along the vertices between $v_i$ and $v_a$. We need to pass at most $|f^{-1}(w_j)|-1$ vertices to enter the set $f^{-1}(w_{j+2})$. We then continue until $f^{-1}(w_{2k-1})$ or $f^{-1}(w_{2k})$ depending on the parity of $j$. Our walk continues from $f^{-1}(w_{0})$ or $f^{-1}(w_1)$ up until we come to the last vertex in $f^{-1}(w_{j-3})$. Finally we take one last step into $f^{-1}(w_{j-1})$ and reach $v_a$. We have then passed
\[
\delta(v_i,v_a) \leq |f^{-1}(w_j)|-1+|f^{-1}(w_{j+2})|+|f^{-1}(w_{j+4})|+\ldots+|f^{-1}(w_{j-3})|+1
\]
vertices. There are $k$ sets among $f^{-1}(w_{j}),\ldots,f^{-1}(w_{j-3})$. When $m \leq k$ each set has either $n$ or $n-1$ vertices. However, at most $\lceil\frac{2(k-m)+1}{2}\rceil =k-m+1$ of them can contain $n$ vertices. Thus,
\[
\delta(v_i,v_a) \leq k(n-1)+ k-m+1 = kn-m+1.
\]
In the case of $m > k$, each set has either $n-1$ or $n-2$ vertices but at most $\frac{2(2k-m+1)}{2}=2k-m+1$ of them can contain $n-1$ vertices. Thus,
\[
\delta(v_i,v_a) \leq k(n-2) + 2k-m+1 = kn-m+1.
\]
When bounding $\delta(v_i,v_b)$ from below, we take a similar walk, but now we want to determine the fewest possible vertices we will pass. Therefore, we assume that we immediately move into the set $f^{-1}(w_{j+2})$ and will go all the way to the last vertex in $f^{-1}(w_{j+1})$. We have then passed a total of
\[
\delta(v_i,v_b) \geq |f^{-1}(w_{j+2})|+|f^{-1}(w_{j+4})|+\ldots+|f^{-1}(w_{j+1})|
\]
vertices. There are $k+1$ sets among $f^{-1}(w_{j+2}),\ldots,f^{-1}(w_{j+1})$.  When $m \leq k$ at least $\frac{2(k-m)}{2}=k-m$ of the sets has $n$ vertices. Thus,
\[
\delta(v_i,v_b) \geq (k+1)(n-1)+k-m = kn-m+n-1.
\]
In the case of $m > k$, at least $\frac{2(2k-m+1)}{2} = 2k-m+1$ has $n-1$ vertices. Thus,
\[
\delta(v_i,v_b)\geq (k+1)(n-2) + 2k-m+1 = kn-m+n-1.
\]
Combining the lower and upper bounds with (\ref{eq:deltail}), we find that
\[
\delta(v_i,v_a) \leq kn-m+1 \leq \delta(v_i,v_l) \leq kn-m+n-1 \leq \delta(v_i,v_b),
\]
hence $f(v_i) f(v_l) \in E(C_{2k+1})$.
Since $v_i v_l$ was an arbitrary edge in $A_2 \cup \cdots \cup A_{\lceil \frac{n+1}{2} \rceil}$, this implies that $f_j = |A_j|$ for $j > 1$.

It remains to determine $f_1$.
Recall that $A_1 = \{ v_i v_{i + kn-m} \,|\, 0 \leq i < 2(kn-m)+n \}$.
As before, we want to check if $\delta(v_i,v_a) \leq \delta(v_i,v_l) = kn-m < kn-m+n-1 \leq \delta(v_i,v_b)$ to determine if $f(v_i) f(v_l) \in E(C_{2k+1})$.
This means that $f(v_i) f(v_l)$ is a non-edge in $C_{2k+1}$ 
if and only if $\delta(v_i,v_a) = kn-m+1$.
%We want to know for how many vertices $v_i \in V(K_{\frac{2kn+n-2m}{kn-m}})$ we have $\delta(v_i,v_a) = kn-m+1$. Since if $\delta(v_i,v_a) < kn-m+1$ then we have $f(v_i) f(v_i+kn-m) \in E(C_{2k+1})$. 
This, in turn, can only happen if the walk from $v_i$ to $v_a$ passes all $|f^{-1}(w_j)|-1$ of the vertices from $f^{-1}(w_j)$ (excluding $v_i$).
Thus, $v_i$ has to be the vertex with the lowest index in $f^{-1}(w_j)$.
In total there are $2k+1$ such vertices, one for each vertex in $C_{2k+1}$.
Furthermore, it must be the case that the walk fully passes the $k-m+1$ sets
$f^{-1}(w_0), f^{-1}(w_2), \ldots, f^{-1}(w_{2k-2m})$
with $n$ vertices in the case when $m \leq k$ and the $2k-m+1$ sets
$f^{-1}(w_0), f^{-1}(w_2), \ldots, f^{-1}(w_{2(2k-m)})$
with $n-1$ vertices when $m > k$.
When $m \leq k$ this happens precisely when $j$ is odd and 
$2(k-m)+3 \leq j \leq 2k-1$, i.e. $m$ times.
When $m > k$ it happens precisely when $j$ is odd and 
$2(k-m+1)+1 \leq j \leq 2k-1$ which is also $m$ times.
In all cases, there will be $m$ edges in $A_1$ which are not mapped to edges in
$E(C_{2k+1})$ so $f_1 = |A_1|-m$ which concludes the proof.
%Since $kn+n-m-1 = 2kn+n-2m -(kn-m+1)$ we have that for all $v_i$ with $f(v_i) = w_j$ at least all nodes between $v_{i+kn-m+1}$ and $v_{i-(kn-m+1)}$ are mapped to either $w_{j-1}$ or $w_{j+1}$ so all edges from $A_2,\ldots,A_r$ are mapped to an edge in $C_{2k+1}$
%\\ \\
%For $A_1$ we see above for all $v_i \neq \min{\{f^{-1}(w_j)\}}$ then $f(v_{i+kn-m}) = w_{j-1}$ or $w_{j+1}$, when $v_i = \min {\{f^{-1}(w_j)\}}$, we have that $f(v_{i+kn-m}) = w_{j-3}$ if and only if
%\begin{equation*}
%|f^{-1}(w_j)|+|f^{-1}(w_{j+2})|+\ldots+|f^{-1}(w_{j-3})|=kn-m+1
%\end{equation*}
%Which for $m \leq k$ is when $k-m+1$ of the nodes $w_j,w_{j+2},\ldots,w_{2{j-3}}$ has $|f^{-1}(w)|=n$, which is the case only if $|f^{-1}(w_{j-1})|=n-1$ and $w_0$ is one of the nodes $w_j,w_{j+2},\ldots,w_{2{j-3}}$. Since when $w_{j-1} \neq w_0$, $w_0$ is among $w_j,w_{j+2},\ldots,w_{2{j-3}}$ exactly half of the time. And $|f^{-1}(w_{j-1})|=n-1$ for $2m$ nodes so 
%\begin{equation*}
%|f^{-1}(w_j)|+|f^{-1}(w_{j+2})|+\ldots+|f^{-1}(w_{j-3})|=kn-m+1
%\end{equation*}
%for $m$ nodes.
%\\ \\
%Similarly when $k > m$ we have that $2k-m+1$ of the nodes $w_j,w_{j+2},\ldots,w_{2{j-3}}$ has $|f^{-1}(w)|=n-1$, which is the case if either $|f^{-1}(w_{j-1})|=n-2$ or if $|f^{-1}(w_{j-1})|=n-1$ and $j-1$ is even. First thing happen in $2m-2k-1$ cases, and second thing in $\frac{4k-2m+2}{2}=2k-m+1$ cases and $2m-2k-1+2k-m+1=m$. So in all cases but $m$ nodes. 
\end{proof}

Let $p$ and $q$ be relatively prime and let $V(K_{p/q}) = \{v_0, \ldots, v_{p-1}\}$. Define a function $\tau : [p] \rightarrow [p]$ by letting $\tau(i) = j$ if $0 \leq j < p$ and $jq \equiv i $ (mod $p$). Note that $\tau$ is a bijection on $[p]$. We will think of $\tau$ as indicating the length of a path (in the positive direction) from $v_0$ to $v_j$ in the cycle $A_1$.
We will denote the length from $v_k$ to $v_l$ in $A_1$ by $\delta_{\tau}(v_k,v_l) = \tau(l)-\tau(k)$ taken modulo $p$. 
Note that $\delta_{\tau}(v_i,v_{i+a}) = \delta_{\tau}(v_0,v_a)$ for all integers $i$.
Closed and half-open intervals are defined by $[v_a,v_b]_{\tau} := \{v_l \;|\; \delta_{\tau}(v_a,v_l) \leq \delta_{\tau}(v_a,v_b)\}$ and $(v_a,v_b]_{\tau} := \{v_l \;|\; 0 < \delta_{\tau}(v_a,v_l) \leq \delta_{\tau}(v_a,v_b)\}$, respectively.

Let $V(C_{2k+1}) = \{w_0, \ldots, w_{2k}\}$.
Given a subset $S \subseteq \{v_0, \ldots, v_{p-1}\}$, we will now describe a general construction
of a solution $f = f_S$ to an instance $(K_{p/q}, \omega)$ of {\sc Max $C_{2k+1}$-COL}.
The idea is to map the nodes $v_{\tau(i)}$ in order of increasing $i$ starting by $f(v_{\tau(0)}) = f(v_0) = w_0$. We then map $v_{\tau(i)}$ to a node adjacent to $f(v_{\tau(i-1)})$, picking one of the two possibilities depending on whether $i+1 \in S$ or not.
To give the formal definition, it will be convenient to introduce the rotation
$\rho$ on $C_{2k+1}$ defined as $\sigma(w_i) = w_{i+1}$.
We then have,
\[
f(v_{\tau(i)}) = \begin{cases}
  w_0 & \text{when $i = 0$,} \\
  \rho^{-1}(f(v_{\tau(i-1)})) & \text{when $i > 0$ and $v_i \in S$,} \\
  \rho(f(v_{\tau(i-1)})) & \text{when $i > 0$ and $v_i \not\in S$.} \\
\end{cases}
\]
Note that the last vertex to be mapped is $v_{\tau^{-1}(p-1)} = v_{p-q}$.
If the created solution has $f(v_{p-q}) = w_1$ or $w_{2k}$, then
$f_1 = |A_1|$, otherwise $f_1 = |A_1| - 1$.
In the latter case, it does not matter whether $v_0 \in S$ or not and we
can assume that $v_0 \not\in S$.
However, to maintain consistency in the case of $f_1 =|A_1|$,
we want to have $v_0 \in S$ if $f(v_{p-q}) = w_1$ and
$v_0 \not\in S$ otherwise.
Therefore, $v_0 \in S$ if and only if $f(v_{p-q}) = w_1$.
%Also notice that a solution with the same signature as the one in Lemma~\ref{lem:solalpha} can be constructed with $S = \emptyset$, although the solution will be mirrored.  

\begin{exmp}
The solution $f : V(K_{22/9}) \rightarrow V(C_5)$ with $S = \{v_{14}$, $v_1$, $v_{11}$, $v_{20}$, $v_7$, $v_{17}$, $v_4$, $v_{13}\}$ looks as follows.
\[
\begin{array}{c c c c c}
{\bf f^{-1}(w_0)} & {\bf f^{-1}(w_1)} & {\bf f^{-1}(w_2)} & {\bf f^{-1}(w_3)} & {\bf f^{-1}(w_4)} \\ \hline
v_0 & v_9 & v_{18} & v_5 & \\
& v_{1} & v_{14} &  \\
& & v_{10} & v_{19} & v_6 \\
v_{15} & v_{2} & & & \\
v_{11} & & & & \\
& & & v_7 & v_{20} \\
& & & & v_{16} \\
v_3 & v_{12} & v_{21} & v_8 & \\
v_{13} & v_4 & v_{17} & &
\end{array}
\]
Note that the $v_i$ are mapped in the order $v_0, v_9, v_{18}, v_5, \ldots, v_{13}$.
$S$ is given in the order in which the vertices appear along the $A_1$.
To start, we let $f(v_0) = w_0$. Neither of $v_9, v_{18}$ or $v_5$ appear in
$S$, so these are mapped consecutively.
Then, we get to $v_{14}$ which is in $S$.
Since $f(v_5) = w_3$ we let $f(v_{14}) = w_2$. 
Finally, $f(v_{13})=w_0$ so the signature of $f$ has $f_1 = |A_1|-1 = 21$.   
\end{exmp}

We will now give some basic properties of the solutions created using this
construction for the case when $p = 2(kn-m)+n$ and $q = kn-m$.
We will from now on assume that $f_1 = |A_1|$. 
This occurs when the construction has an equal number of applications of
$\rho$ and $\rho^{-1}$ modulo $2k+1$. That is, when $|S| \equiv p-|S| $ (mod $2k+1$).
Solving for $|S|$ we get:
\begin{equation} \label{eq:fulla1}
  |S| \equiv 2k+1-m \text{ (mod $2k+1$)}.
\end{equation}
%This property has some useful consequences.
Assume that $f(v_i) = w_j, f(v_{i'}) = w_{j'}$.
Then, the index $j'$ is determined by $\delta_{\tau}(v_i,v_{i'})$ and
$S \cap (v_i, v_{i'}]_{\tau}$ as follows:
\begin{equation} \label{eq:index}
  j' \equiv j+\delta_{\tau}(v_i,v_{i'}) - 2 \cdot |S \cap (v_i, v_{i'}]_{\tau}| \text{ (mod $2k+1$)}.
\end{equation}
%In particular,
%\begin{itemize}
%\item
%  If $S \cap (v_i,v_{i-m}]_{\tau}=\emptyset$, then $f(v_i)=f(v_{i-m})$.
%\item
%  $f(v_i) f(v_{i'}) \in C_{2k+1}$ if and only if
%  $|S \cap (v_i, v_{i'}]_{\tau}| \equiv (k+1)(\delta_{\tau}(v_i,v_{i'}) \pm 1) $ (mod $2k+1$).
%\end{itemize}
  
%\begin{proof}
%  For the first property, let $i' = i-m$.
%  Then, (\ref{eq:index}) reads
%  $j' \equiv j+\delta_{\tau}(v_i,v_{i-m})-2 \cdot 0 $ (mod $2k+1$).
%  Since $(2k+1)q \equiv -m $ (mod $p$), we have $\delta_{\tau}(v_i,v_{i-m}) \equiv 0 $ (mod $2k+1$), so $j' \equiv j $ (mod $2k+1$) and $f(v_i) = f(v_{i-m})$.

%  For the second property, note that we want $j' \equiv j \pm 1 $ (mod $2k+1$)
%  in (\ref{eq:index}). Rearranging, we have
%  \[
%  2 \cdot |S \cap (v_i, v_{i'}]_{\tau}| \equiv \delta_{\tau}(v_i,v_{i'}) \pm 1 \text{ (mod $2k+1$)}.
%  \]
%  Multiplication by $k+1$ yields the desired result.
%\end{proof}

\noindent
Relation (\ref{eq:index}) implies the following useful lemma:

\begin{lemma} \label{lem:usefulcong}
  $f(v_i) f(v_{i'}) \in E(C_{2k+1})$ iff
  $|S \cap (v_i, v_{i'}]_{\tau}| \equiv (k+1)(\delta_{\tau}(v_i,v_{i'}) \pm 1) $ (mod $2k+1$).
\end{lemma}

\begin{lemma}
\label{lem:solbeta}
Let $k,n \geq 2$ be integers.
There exists a solution $f$ to $(K_{\frac{2kn+n-2}{kn-1}},\omega)$ of {\sc Max $C_{2k+1}$-COL} with $f_1 = |A_1|$, and
%If in addition $f_2$ is maximised, then $f$ can be chosen so that
\[
f = \begin{cases}
(|A_1|,|A_2|-2(2k-1),\ldots,|A_{\frac{n+1}{2}}|-2(2k-1)) & \text{if $n$ is odd,} \\
(|A_1|,|A_2|-2(2k-1),\ldots,|A_{\frac{n}{2}}|-2(2k-1),|A_{\frac{n+2}{2}}|-(2k-1)) & \text{if $n$ is even.}
\end{cases}
\]
Furthermore, for any other solution $g$, if $g_1 = |A_1|$, then $g_2 \leq f_2$,
componentwise.
\end{lemma}

\begin{proof}
Let $p = 2(kn-1)+n, q = kn-1$, and $V(K_{p/q}) = \{v_0,\ldots,v_{p-1}\}$.
The desired solution $f$ is obtained from the construction $f = f(S)$
with $S = [v_{\tau^{-1}(p-2k+1)},v_0]_{\tau}$.
As required by (\ref{eq:fulla1}), we have $|S| = p-(p-2k+1)+1 = 2k$ so that
$f_1 = |A_1|$.
It remains to determine $f_c$ for $c > 1$.

Let $v_i v_{i'} \in A_c, c > 1$ be an edge.
In order to count the edges only once, we will assume that $i' = i+q+(c-1) $ (mod $p$). %assume that $\tau(i) < \tau(i')$.
% It is easy to check that this implies $v_i \not\in S \setminus \{v_0\}$.
To be able to use the condition in Lemma~\ref{lem:usefulcong} we need to
determine $\delta_{\tau}(v_i, v_{i'})$.
But, $\tau(i') \equiv \tau(i)+1+q^{-1}(c-1) $ (mod $p$), where
$q^{-1} := p-2k-1$, the inverse of $q$ modulo $p$.
We then obtain $\delta_{\tau}(v_i, v_{i'}) = \tau(i')-\tau(i)$ by reducing $1+(p-2k-1)(c-1)$ modulo $p$.
\[
\delta_{\tau}(v_i, v_{i'}) = \begin{cases}
  1 & \text{if $c = 1$, and} \\
  -1+(2k+1)(n-c+1) & \text{otherwise.} \\
\end{cases}
\]
Assuming $c \neq 1$, we have two cases in Lemma~\ref{lem:usefulcong}.
We conclude that 
$f(v_i) f(v_{i'}) \in E(C_{2k+1})$ if and only if
either 
\begin{equation}
\label{eq:recycle}
|S \cap (v_{i},v_{i'}]_{\tau}| \equiv 0 \qquad \text{or} \qquad
|S \cap (v_{i},v_{i'}]_{\tau}| \equiv (k+1)(-2) \equiv 2k \text{ (mod $2k+1$)}.
\end{equation}
In both cases the condition is equivalent to $v_i, v_{i'} \not\in S \setminus \{v_0\}$.
Therefore, the edges $v_i v_i'$ which are not mapped to an edge in $C_{2k+1}$ by $f$ are the ones with an endpoint in $S \setminus \{v_0\}$. (There are no edges with both endpoints in this set.) 
When $n$ is even and $c = n/2+1$, this number equals $|S \setminus \{v_0\}| = 2k-1$.
In all other cases, there are $2(2k-1)$ such edges.
The first part of the lemma follows.

For the second part, we pick an arbitrary solution $g$ and show that we can
find at least $2(2k-1)$ edges in $A_2$ which can not be mapped to $C_{2k+1}$,
provided that $g_1 = |A_1|$.
It is easy to see that, up to rotational symmetry, a $g$ with $g_1 = |A_1|$ 
must be constructible by $g = g(S)$ for some $S$. 
We already know that such an $S$ must satisfy $|S| \equiv 2k$ (mod $2k+1$).
This implies $|S| \geq 2k$.
From $p \equiv 2k-1$ (mod $2k+1$), we also see that we must have
|$V(K_{p/q}) \setminus S| \geq 2k$.
As argued before, an edge from $A_c$ is mapped to $C_{2k+1}$ if and only if one of the
congruences in (\ref{eq:recycle}) holds.
Since $|S| \equiv 2k$ (mod $2k+1$), we can equivalently write this as
$f(v_i) f(v_{i'}) \in E(C_{2k+1})$ if and only if
either 
\begin{equation}
\label{eq:recycle2}
|S \cap (v_{i'},v_{i}]_{\tau}| \equiv 2k \qquad \text{or} \qquad
|S \cap (v_{i'},v_{i}]_{\tau}| \equiv 0 \text{ (mod $2k+1$)}.
\end{equation}
Hence, either the intersection of $S$ with $(v_{i'},v_{i}]_{\tau}$ is empty or
the latter is a subset of the former.
As the two cases can be treated identically, we assume, 
without loss of generality, that the intersection is empty.
Note that $|(v_{i'},v_{i}]_{\tau}| = 2k$
We will now determine $2(2k-1)$ edges which can not be mapped to edges
in $C_{2k+1}$.
Let $v_{j_{1}}$ be the first vertex in $S$ encountered following $A_1$ from $v_i$ in
the positive direction.
Similarly, let $v_{j_{2}}$ be the first vertex in $S$ encountered following $A_1$ 
from $v_{i'}$ in the negative direction.
Then, $v_{j_{1}}, v_{j_{1}-q} \in (v_{j_{1}+(a+1)q+1}, v_{j_{1}+aq}]_{\tau}$,
for $a = 0, \ldots, 2k-2$,
but $v_{j_{1}} \in S$ and $v_{j_{1}-q} \not\in S$ by construction.
Thus, from (\ref{eq:recycle2}), the edges $v_{j_{1}+aq} v_{j_{1} +(a+1)q+1}$ 
can not be mapped to $C_{2k+1}$.
In the other direction, we have $v_{j_{2}}, v_{j_{2}+q} \in (v_{j_{2}-a'q}, v_{j_{2}+(1-a')q+1}]_{\tau}$,
for $a' = 0, \ldots, 2k-2$,
but $v_{j_{2}} \in S$ and $v_{j_{2}+q} \not\in S$ by construction.
From this we get another $2k-1$ edges which can not be mapped to $C_{2k+1}$.
Finally, we note that since $S \subseteq [v_{j_1}, v_{j_2}]_{\tau}$ and
$|S| \geq 2k$, the edges
$v_{j_{1}+aq} v_{j_{1} +(a+1)q+1}$ and
$v_{j_{2}+(1-a')q+1} v_{j_{2}-a'q}$ are distinct.
This proves that $g_2 \leq f_2$.
\end{proof}

\begin{exmp}
With $k=3$ and $n=5$ the solution $f = f(S)$ to$ (K_{33/14},\omega)$ of MAX $C_7$-COL created as in Lemma~\ref{lem:solbeta} with $S = \{v_{29}$, $v_{10}$, $v_{24}$, $v_{5}$, $v_{19}$, $v_{0}\}$ looks like:
\[
\begin{array}{c c c c c c c}
{\bf f^{-1}(w_0)} & {\bf f^{-1}(w_1)} & {\bf f^{-1}(w_2)} & {\bf f^{-1}(w_3)} & {\bf f^{-1}(w_4)} & {\bf f^{-1}(w_5)} & {\bf f^{-1}(w_6)} \\ \hline
v_0 & v_{14} & v_{28} & v_9 & v_{23} & v_4 & v_{18} \\
v_{32} & v_{13} & v_{27} & v_8 & v_{22} & v_3 & v_{17} \\
v_{31} & v_{12} & v_{26} & v_7 & v_{21} & v_2 & v_{16} \\
v_{30} & v_{11} & v_{25} & v_6 & v_{20} & v_1 & v_{15} \\
& v_{19} & v_5 & v_{24} & v_{10} & v_{29} &
\end{array}
\]
\end{exmp}

\section{Proof of Proposition~\ref{th:qisfunny}}

The proof of Proposition~\ref{th:qisfunny} follows from a series of lemmas.
The function $\delta_{\tau}$ and how it is used for constructing solutions is
presented in Appendix~\ref{app:propm1proof}.

\begin{lemma}
\label{lem:splitend}
Let $k \geq 2$ be an integer, and $n \geq 3$ be an odd integer. Then, there exists a solution $f$ to $(K_{\frac{2kn+n-4}{kn-2}},\omega)$ of MAX $C_{2k+1}$-COL with the following signature:
\begin{equation}
\begin{array}{llll}
  f_1 & = & |A_1|,\\
  f_{2i} & = & |A_{2i}|-(\frac{n-1}{2}-i)(2k+1)-(4k-2) & \text{for $i = 1, 2, \ldots, \frac{n+1}{4}$}, \\
  f_{2i+1} & = & |A_{2i+1}|-(i-1)(2k+1)-(4k-2) & \text{for $i = 1,2,\ldots,\frac{n-1}{4}$}.
\end{array}
\end{equation}
\end{lemma}
\begin{proof}
Let  $G = K_{\frac{2kn+n-4}{kn-2}}$, $V(G) =\{ v_0,\ldots,v_{2kn+n-5}\}$ and $V(C_{2k+1}) = \{w_0,\ldots,w_{2k}\}$. A solution $f=f(S)$ with this signature is obtained with
$S = [v_{\tau^{-1}(2kn+n-2k-2)},v_0]_{\tau}.$
We then have $|S| = (2kn+n-4) - (2kn+n-2k-2) + 1 = 2k-1$, so $f_1 = |A_1|$ by (\ref{eq:fulla1}).

An orbit $A_{2i}$ includes edges which connects vertices at a distance $kn-2+2i-1$. $\delta_{\tau}(v_j,v_{j+kn+2i-3}) = (\frac{n-1}{2}-(i-1))(2k+1)-1$. 
Lemma~\ref{lem:usefulcong} then says that $f(v_j) f(v_{j+kn+2i-3}) \in E(C_{2k+1})$
if and only if $|S \cap (v_{j},v_{j+kn+2i-3}]_{\tau}| \equiv 0$ or $2k$ (mod $2k+1$).
That is, $|S \cap (v_{j},v_{j+kn+2i-3}]_{\tau}|$ must be 0.
This is the case only when
\[
\delta_{\tau}(v_0,v_j) \leq \delta_{\tau}(v_0,v_{j+kn+2i-3}) < \delta_{\tau}(v_0,v_{\tau^{-1}(2kn+n-2k-2)}),
\]
which implies
\begin{multline*}
\delta_{\tau}(v_0,v_j) \leq \delta_{\tau}(v_0,v_{\tau^{-1}(2kn+n-2k-2)}) - \delta_{\tau}(v_j,v_{j+kn+2i-3}) - 1 \\
\leq  2kn+n-2k-2-((\frac{n-1}{2}-(i-1))(2k+1)-1)-1 \\
= 2kn+n-4 -(\frac{n-1}{2}-i)(2k+1)-(4k-2)-1,
\end{multline*}
which holds for exactly $2kn+n-4 -(\frac{n-1}{2}-i)(2k+1)-(4k-2)$ vertices $v_j$.

An orbit $A_{2i+1}$ includes edges which connects vertices at a distance $kn-2+2i$. $\delta_{\tau}(v_{j+kn+2i-2},v_j) = i(2k+1)-1$. Applying Lemma~\ref{lem:usefulcong} again asserts that $S \cap (v_{j},v_{j+kn+2i-2}]_{\tau}$ must be empty. Thus,
\[
\delta_{\tau}(v_0,v_{j+kn+2i-2}) \leq \delta_{\tau}(v_0,v_j) < \delta_{\tau}(v_0,v_{\tau^{-1}(2kn+n-2k-2)}),
\]
which implies
\begin{multline*}
\delta_{\tau}(v_0,v_j) \leq \delta_{\tau}(v_0,v_{\tau^{-1}(2kn+n-2k-2)}) - \delta_{\tau}(v_j,v_{j+kn+2i-2}) - 1 \\
\leq 2kn+n-2k-2-(i(2k+1)-1) -1 \\
= 2kn+n-4-(i-1)(2k+1)-(4k-2) - 1,
\end{multline*}
which holds for exactly $2kn+n-4-(i-1)(2k+1)-(4k-2)$ vertices $v_j$.
\end{proof} 

\begin{exmp}
For $K_{31/13}$. The solution $f=f(S)$ as in Lemma~\ref{lem:splitend} with $k=3$ and $n=5$ has $S = \{v_{10}$, $v_{23}$, $v_{5}$, $v_{18}$, $v_{0}\}$ and looks like:
\[
\begin{array}{c c c c c c c}
{\bf f^{-1}(w_0)} & {\bf f^{-1}(w_1)} & {\bf f^{-1}(w_2)} & {\bf f^{-1}(w_3)} & {\bf f^{-1}(w_4)} & {\bf f^{-1}(w_5)} & {\bf f^{-1}(w_6)} \\ \hline
v_0 & v_{13} & v_{26} & v_8 & v_{21} & v_3 & v_{16} \\
v_{29} & v_{11} & v_{24} & v_6 & v_{19} & v_1 & v_{14} \\
v_{27} & v_{9} & v_{22} & v_4 & v_{17} & v_{30} & v_{12} \\
v_{25} & v_{7} & v_{20} & v_2 & v_{15} & v_{28} &  \\
& v_{18} & v_5 & v_{23} & v_{10} & &
\end{array}
\]
\end{exmp}

The following technical lemma will prove useful in analysing the solutions
in Lemma~\ref{lem:corw}.
Some cases of the defined (partial) function $\gamma$ which are not needed 
for this analysis have been left out.
%Before presenting the next set of signatures we present a lemma that will make it easier for us to see how many edges from the different orbits are mapped to $C_{2k+1}$ by a solution $f(S)$ if $S$ has a certain form. This next Lemma will be very specific for our purpose, but the function introduced can with a little work quite easily be extended to cover all possible cases. 

\begin{lemma}
\label{lem:gamma}
Let $p,q,r,s$ be positive integers so that $r > 2p+q+s$ and $s \geq p$. Now consider $r$ elements equidistantly placed on a circle, and select two sequences $P_1$ and $P_2$, each containing $p$ consecutive elements, with $q$ elements between them on one side and $r-2p-q$ on the other side. Let $\gamma(i)$ be the number of ways to select $s$ consecutive elements on the circle with exactly $i$ elements from $P_1 \cup P_2$. Then,
when $s \leq q$:
\[
\gamma(i) = \begin{cases}
  r-2p-2s+2 & \text{if $i = 0$,} \\
  2s-2p+2 & \text{if $i = p$,} \\
  0 & \text{if $i > p$.}
\end{cases}
\]
%$\gamma(0) = r-2p-2s+2$ \\
%$\gamma(1)=\cdots=\gamma(p-1)=4$ \\
%$\gamma(p) = 2s-2p+2$ \\
%$\gamma(i) = 0 \;\forall \; i > p$,
when $s = q+p$:
\[
\gamma(i) = \begin{cases}
  r-3p-2q+1 & \text{if $i = 0$,} \\
  2q+p+1 & \text{if $i = p$,} \\
  0 & \text{if $i > p$.}
\end{cases}
\]
%$\gamma(0) = r-3p-2q+1$ \\
%$\gamma(p) = 2q+p+1$ \\
%$\gamma(i) = 0 \; \forall \; i > p$,
and when $s > q+p+1$:
\[
\gamma(i) = \begin{cases}
  r-2p-q-s+1 & \text{if $i = 0$,} \\
  2q+2 & \text{if $i = p$,} \\
  2 & \text{if $i = p+1$,} \\
  0 & \text{if $i > 2p$.}
\end{cases}
\]
%$\gamma(0) = r-2p-q-s+1$ \\
%$\gamma(p) = 2q+2$ \\
%$\gamma(p+1) = 2$ \\
%$\gamma(i) = 0 \; \forall \; i > 2p$.
\end{lemma}
\begin{proof}
Call the elements $\{0,\ldots,r-1\}$. Suppose $P_1 = \{0,\ldots,p-1\}$ and $P_2 = \{q+p,\ldots,q+2p-1\}$. 
For $s \leq q$. Then the sequences that starts with $p,\ldots,q+p-s$ as well as $q+2p,\ldots,r-s$ are the only ones that do not contain any element from $P_1 \cup P_2$ and that is $q+p+s+1$ and $r-q-2p-s+1$ elements, and in total $r-2p-2s+2$. To get $p$ elements we have the sequences that starts with $q-2p-s,\ldots,q+p$ and $p-s,\ldots,0$ as the only options, and that is $s-p+1$ in both cases so $2s-2p+2$ in total. Also if $s \leq q$ clearly there is no sequence of length $s$ that includes element from both $P_1$ and $P_2$.

For $s=q+p$, the ones starting with $q+2p,\ldots,r-s$ are the only ones that do not contains any element from $P_1 \cup P_2$ and that is $r-q-2p-s+1=r-3p-2q+1$ elements. To get $p$ elements we have that for any sequence of length $s$ starting with an element $i \in P_1$ contains $i$ number of elements from $P_2$ so all sequences starting with $r-q,\ldots,q+p$ contains $p$ elements from $P_1 \cup P_2$, and that is $p+2q+1$ in total. This fact also makes it clear that no sequence can contain more than p elements from $P_1 \cup P_2$.

Finally when $s > q+p+1$, again sequences starting with $q+2p,r-s$ are the only ones to not contain any element from $P_1$ or $P_2$, and for sequences that include $p$ element, they must start with $p,\ldots,q+p$ or $r+p-s,\ldots,r+p+q-s$. That is both with $q+1$ for a total of $2q+2$. Also the only ones to contain $p+1$ elements are the ones starting with $p-1$ and $r+p+q-s+1$, and since $|P_1 \cup P_2| = 2p$ there are of course no sequence to contain more than that.
\end{proof}

\begin{lemma}
\label{lem:corw}
Let $f=f(S)$ be a solution to $(K_{\frac{2kn+n-2m}{kn-m}},\omega)$ of {\sc MAX $C_{2k+1}$-COL} with $f_1 = |A_1|$ and where $S = P_1 \cup P_2$, and $P_1 \cap P_2 = \emptyset$, where
%\begin{equation*}
$P_1 = [v_a,v_{a+(2k+1)-2}]_{\tau}$
%\end{equation*}
and 
%\begin{equation*}
$P_2 = [v_b,v_{b+(2k+1)-2}]_{\tau}$
%\end{equation*}
such that 
%\begin{equation*}
$\min_{v_i\in P_1, v_j \in P_2}{\bar{\delta}_{\tau}(v_i,v_j)} = (u-1)(2k+1)$.
%\end{equation*}
Let $A_c$ be the orbit consisting of edges $v_l v_h$ with $\delta_{\tau}(v_l,v_h)=g(2k+1)-1$. Then,
\[
f_c = \begin{cases}
  |A_c|-(8k-4) & \text{if $g < u$,} \\
  |A_c|-(4k-2) & \text{if $g = u$,} \\
  |A_c|-(g-u)(2k+1)-(6k-5) & \text{if $g > u$.}
\end{cases}
\]
\end{lemma}
\begin{proof}
We can apply Lemma~\ref{lem:gamma}, since we according to Lemma~\ref{lem:usefulcong} must have $|S \cap [v_{l+1},v_h]_{\tau}|=0,2k$ or $2k+1$. So for Lemma~\ref{lem:gamma} we have $r=|A_c|$, $p=2k$, $q=(u-1)(2k+1)$ and $s=g(2k+1)-1$. We see that $s=p+q$ when $g=u$ and when $g<u$ then $s\leq q$ and when $g > u$ then $s > q+p+1$. So all we have to do is for each case count $\gamma(0)+\gamma(p)+\gamma(p+1)$.
\end{proof}

Now it is possible to construct a series of signatures with solutions $f(S)$ where $S$ will have the properties sought after by Lemma~\ref{lem:corw}.

\begin{lemma}
\label{lem:splitmiddle}
There exists a set of solutions $F = \{ f^i \}, i = 2,\ldots,\frac{n+1}{2}$ to $(K_{\frac{2kn+n-4}{kn-2}},\omega)$ of the problem $\MCol{C_{2k+1}}$ with signatures:
\begin{eqnarray*}
  & & f^i_1=|A_1| 
  \; \forall f^i \in F \\
  & & f^i_i=|A_i|-(4k-2) 
  \; \forall f^i \in F \\
  & & f^{2i}_{2j+1}=|A_{2j+1}|-(8k-4)
  \; \forall f^{2i} \in F \mbox{ and } j = 1,2,\ldots,\frac{n-1}{4} \\
  & & f^{2i}_{2j}=|A_{2j}|-(8k-4)
  \; \forall f^{2i} \in F \mbox{ and }  j = i+1,i+2,\ldots,\frac{n+1}{4} \\
  & & f^{2i+1}_{2j+1}=|A_{2j+1}|-(8k-4)
  \; \forall f^{2i+1} \in F \mbox{ and } j = 1,2,\ldots,i-1 \\
  & & f^{2i}_{2j}=|A_{2j}|-(i-j)(2k+1)-(6k-5)
  \; \forall f^{2i} \in F \mbox{ and } j = 1,2,\ldots,i-1 \\
  & & f^{2i+1}_{2j}=|A_{2j}|-(\frac{n-1}{2}-i-j)(2k+1)-(6k-5) 
  \; \forall f^{2i+1} \in F \mbox{ and } j = 1,2,\ldots,\frac{n+1}{4} \\
  & & f^{2i+1}_{2j+1}=|A_{2j+1}|-(j-i)(2k+1)-(6k-5)
  \; \forall f^{2i+1} \in F \mbox{ and } j = i+1,i+2,\ldots,\frac{n-1}{4} \\
\end{eqnarray*}
%\begin{align*}
%f^i_1&=|A_1| &&\forall \; f^i \in F \\
%f^i_i&=|A_i|-(4k-2) &&\forall \; f^i \in F \\
%f^{2i}_{2j+1}&=|A_{2j+1}|-(8k-4) &&\for
%all \; f^{2i} \in F \mbox{ and } 1 < j \leq \frac{n-1}{4}  \\
%f^{2i}_{2j}&=|A_{2j}|-(8k-4) &&\forall \; f^{2i} \in F \mbox{ and } i < j \leq \frac{n+1}{4} \\
%f^{2i+1}_{2j+1}&=|A_{2j+1}|-(8k-4) &&\forall \; f^{2i} \in F \mbox{ and } 1 \leq j < i  \\
%f^{2i}_{2j}&=|A_{2j}|-(i-j)(2k+1)-(6k-5) &&\forall \; f^{2i} \in F \mbox{ and } 1 < j < i \\ 
%f^{2i+1}_{2j}&=|A_{2j}|-(\frac{n-1}{2}-i-j)(2k+1)-(6k-5) &&\forall \; f^{2i+1} \in F \mbox{ and } 1 < j \leq \frac{n+1}{4} \\
%f^{2i+1}_{2j+1}&=|A_{2j+1}|-(j-i)(2k+1)-(6k-5) &&\forall \; f^{2i+1} \in F \mbox{ and } i < j \leq \frac{n-1}{4}
%\end{align*}
\end{lemma}
\begin{proof}
Let $f^{2i+1} = f(S)$ with $S = P_1 \cup P_2$,
where
$P_1= [v_{\tau^{-1}((n-i-1)(2k+1)-1)},v_{\tau^{-1}((n-i)(2k+1)-3)}]_{\tau}$
and
$P_2 = [v_{\tau^{-1}((n-1)(2k+1)-2)},v_0]_{\tau}$.
We have $|P_1|=|P_2|=2k$ so $|S| = 4k$, implying $f^{2i+1}_1=|A_1|$ due to (\ref{eq:fulla1}).

The orbits $A_{2j+1}$ include edges which connects vertices at a distance $kn-2+2j$ and
$\delta_{\tau}(v_{l+kn+2j-2},v_l) = j(2k+1)-1$. We now have the situation in Lemma~\ref{lem:corw} with $u=i$ and $g=j$. When $i<j$ then $f^{2i+1}_{2j+1}=|A_{2j+1}|-(j-i)(2k+1)-(6k-5)$. When $j=i$ then $f^{2i+1}_{2i+1} = |A_{2i+1}|-(4k-2)$ and when $i>j$ then $f^{2i+1}_{2j+1} = |A_{2j+1}|-(8k-4)$.

The orbits $A_{2j}$ include edges which connects vertices at a distance $kn-2+2j-1$ and 
$\delta_{\tau}(v_l,v_{l+kn+2j-3}) = (\frac{n-1}{2}-j)(2k+1)-1$. Since we have $i \leq \left\lceil\frac{n+1}{4}-1\right\rceil$ and $j \leq \frac{n+1}{4}$, then $\frac{n-1}{2}-j > i$ for all $i$ and $j$. So for Lemma~\ref{lem:corw} only the third case applies with $u=i$ and $g = \frac{n-1}{2}-j$. Thus, we have $f^{2i+1}_{2j} = |A_{2j}|-(\frac{n-1}{2}-i-j)(2k+1)-(6k-5)$.

Let $f^{2i} = f(S)$ with 
$S = P_1 \cup P_2$, where
$P_1= [v_{\tau^{-1}((\frac{n-1}{2}+i-1)(2k+1)-1)},v_{\tau^{-1}((\frac{n-1}{2}+i-1)(2k+1)-3)}]_{\tau}$
and
$P_2 = [v_{\tau^{-1}((n-1)(2k+1)-2)},v_0]_{\tau}$.
Again, we have $|P_1|=|P_2|=2k$ so $|S| = 4k$ and $f^{2i+1}_1=|A_1|$.

For the orbits $A_{2j+1}$ we have that $j < \frac{n-1}{2}-i$ for all $i$ and $j$ so only case 1 in Lemma~\ref{lem:corw} applies and $f^{2i}_{2j+1}=|A_{2j+1}|-(8k-4)$.

For the orbits $A_{2j}$ we have exactly the same situation as in Lemma~\ref{lem:corw} with $u = \frac{n-1}{2}-i$ and $g = \frac{n-1}{2}-j$. We notice that when $i<j$ then $g<u$ and  $f^{2i}_{2j}=|A_{2j}|-(8k-4)$, when $i>j$ then $g>u$ and $f^{2i}_{2j}=|A_{2j}|-(i-j)(2k+1)-(6k-5)$ and when $i=j$ then $g=u$ and $f^{2i}_{2i} = |A_{2i}|-(4k-2)$.
\end{proof}

One important thing to notice here is that $f^{2i}$ with $l= \left\lfloor\frac{n+1}{4}\right\rfloor$ is the continuation of $f^{2j+1}$ with $j = \left\lfloor\frac{n-1}{4}\right\rfloor$. Since $\frac{n-1}{2}-\left\lfloor\frac{n-1}{4}\right\rfloor=\left\lfloor\frac{n+1}{4}\right\rfloor$. Another observation is that signature $f^{3}$ from Lemma~\ref{lem:splitmiddle} can always be removed from a complete set of signatures and it will still remain complete since the signature from Lemma~\ref{lem:splitend} is better or equal for all orbits. 

\begin{exmp}
With $k=3$ and $n=5$ the solution $f= f(S)$ to $K_{31/13}$ of {\sc MAX $C_7$-COL} from Lemma~\ref{lem:splitmiddle} has
\[
S = \{v_{14}, v_{27}, v_{9}, v_{22}, v_{4}, v_{17}\} \cup \{v_{28}, v_{10}, v_{23}, v_{5}, v_{18}, v_{0}\},
\]
and looks like:
\[
\begin{array}{c c c c c c c}
{\bf f^{-1}(w_0)} & {\bf f^{-1}(w_1)} & {\bf f^{-1}(w_2)} & {\bf f^{-1}(w_3)} & {\bf f^{-1}(w_4)} & {\bf f^{-1}(w_5)} & {\bf f^{-1}(w_6)} \\ \hline
v_0 & v_{13} & v_{26} & v_8 & v_{21} & v_3 & v_{16} \\
v_{29} & v_{11} & v_{24} & v_6 & v_{19} & v_1 &  \\
v_4 & v_{22} & v_{9} & v_{27} & v_{14} & & \\
& & & & & & v_{17} \\
v_{30} & v_{12} & v_{25} & v_{7} & v_{20} & v_2 & v_{15}  \\
& v_{18} & v_5 & v_{23} & v_{10} & v_{28} &
\end{array}
\]
\end{exmp}

\noindent
We will now prove Proposition~\ref{th:qisfunny} using the solutions from
Lemma~\ref{lem:solalpha}, Lemma~\ref{lem:splitend} and Lemma~\ref{lem:splitmiddle}.

\subsection*{Proposition~\ref{th:qisfunny}}
\begin{proof}
We get $(n+1)/2$ inequalities from Lemma~\ref{lem:solalpha},~\ref{lem:splitend}, 
and~\ref{lem:splitmiddle}, where as noted above, we have removed the inequality
generated by $f^3$.
As variables we have $s$ and $\omega_i$, for $i = 1, \ldots, (n+1)/2$.
To solve the relaxation of (\ref{lp}), we solve the corresponding system with
equalities.
A similar treatment of the dual confirms that the obtained solution is
indeed the optimum.

%We will reduce the signature from Lemma~\ref{lem:splitmiddle} into one equation to be able to solve a $4\times4$ equation system. We will then construct the dual problem of~(\ref{lp}) and treat it the same way. By showing that these two will provide the same objective value, we can say we have an optimal solution. 

We start by reducing our $\frac{n+1}{2} \times \frac{n+3}{2}$ system to
a $4 \times 4$ system.
However we need to rearrange the orbits to conveniently describe how they depend on each other. Let $A_1'=A_1, A_2'=A_3,\cdots,A_i'=A_{2j+1}, A_{i+1}'=A_{2l},A_{i+2}'=A_{2l-2},\cdots,A_{\frac{n+1}{2}}'=A_2$, where $j = \left\lfloor\frac{n-1}{4}\right\rfloor$ and $l= \left\lfloor\frac{n+1}{4}\right\rfloor$. Furthermore introduce new solutions $h$ so that $h^i$ denotes the solution that maximises $h_i$. This rearrangement makes sense, as it puts the orbits and solution in such an order that for all solutions $h^r$ we have $h^r_r > h^r_{r+1} > h^r_{r+2} > \cdots > h^r_{\frac{n+1}{2}}$.

Now we compare the equations in (\ref{lp}) from the signatures $h^{\frac{n+1}{2}}$ and $h^{\frac{n+1}{2}-1}$.  Note that these are the signatures $f^2$ and $f^4$ from Lemma~\ref{lem:splitend}. We see then that we have 
\[
\sum_j{h^{\frac{n+1}{2}-1}_j \omega_j}=\sum_j{h^{\frac{n+1}{2}}_j \omega_j} + (4k-2)\cdot\omega_{\frac{n+1}{2}-1} - (4k-2)\cdot\omega_{\frac{n+1}{2}},
\] 
since we assume $\displaystyle{\sum_j{h^{\frac{n+1}{2}-1}_j \omega_j}=\sum_j{h^{\frac{n+1}{2}}_j \omega_j}=s}$, we get $\omega_{\frac{n+1}{2}}=\omega_{\frac{n+1}{2}-1}$.
For the general case we have
\begin{equation}
\sum_j{h^i_j \omega_j}=\sum_j{h^{i+1}_j \omega_j} + (4k-2)\cdot\omega_i - (4k-2)\cdot\omega_{i+1} 
- (2k+1)\cdot\sum_{j=i+2}^{\frac{n+1}{2}}\omega_j,
\end{equation}
for $i = 3,4,\ldots,\frac{n+1}{2}-1$.
Since again we assume $\sum_j{h^i_j \omega_j}=\sum_j{h^{i+1}_j \omega_j}=s$, we get
\begin{equation}
\label{eq:omega}
\omega_i=\omega_{i+1}+\frac{2k+1}{4k-2}\cdot\displaystyle\sum_{j=i+2}^{\frac{n+1}{2}}\omega_j,
\end{equation}
for $i = 3,4,\ldots,\frac{n+1}{2}-1$. For $i=\frac{n+1}{2}-1$ this means $\omega_i=\omega_{i+1}$. For all other $i$, we use the fact that (\ref{eq:omega}) also holds for $\omega_{i+1}$ and thus have:
\begin{equation}
\label{eq:omega+1}
\omega_{i+1}=\omega_{i+2}+\frac{2k+1}{4k-2}\cdot\displaystyle\sum_{j=i+3}^{\frac{n+1}{2}}\omega_j,
\end{equation}
for $i = 3,4,\ldots,\frac{n+1}{2}-2$.
From (\ref{eq:omega+1}) we get,
\begin{equation}
\label{eq:sum}
(2k+1)\sum_{j=i+3}^{\frac{n+1}{2}}\omega_j = (4k-2)\cdot(\omega_{i+1}-\omega_{i+2}).
\end{equation}
We then insert (\ref{eq:sum}) into (\ref{eq:omega}) to express $\omega_i$ in terms of $\omega_{i+1}$ and $\omega_{i+2}$ only:
\begin{equation}
\omega_i = \omega_{i+1}+\frac{2k+1}{4k-2} \cdot \omega_{i+2} + (\omega_{i+1}-\omega_{i+2})
= 2\cdot\omega_{i+1}-\frac{2k-3}{4k-2}\cdot\omega_{i+2},
\end{equation}
for $i = 3,4,\ldots,\frac{n+1}{2}-2$.
We now define $\omega_i = g_{\frac{n+1}{2}-i}\cdot\omega_{\frac{n+1}{2}}$ with
\[
g_i =
\begin{cases}
  1 & i = 0, 1, \\
  2 \cdot g_{i-1}-\frac{2k-3}{4k-2}\cdot g_{i-2} & i = 2, 3, \ldots, \frac{n+1}{2}-3.
\end{cases} 
\]
Thus, we can express $\omega_3,\ldots,\omega_{\frac{n+1}{2}-1}$ in terms of $\omega_{\frac{n+1}{2}}$. However, to proceed we need to express the coefficients $g_i$ in terms of $k$ and $n$. Define $G(z) = \sum_{g\geq0}{g_n z^n}$. We do not have to worry about the upper limit, since as far as we are concerned the recursion could go on towards infinity, without affecting the values we are interested in. After multiplying with $z^n$ and summing up from $n \geq 2$ we get
\[
g_2z^2+g_3z^3+\cdots = 2\{g_1z^2+g_2z^3+\cdots\}-\frac{2k-3}{4k-2}\{g_0z^2+g_1z^3+\cdots\},
\]
which we identify as
\[
G(z)-z-1=2z(G(z)-1)-\frac{2k-3}{4k-2}z^2G(z).
\]
Solving for $G(z)$ gives
\[
G(z) = \frac{1-z}{\frac{2k-3}{4k-2}z^2-2z+1}.
\]
The denominator has two distinct roots whose reciprocals are:
\[
\alpha_1=1+\sqrt{1-\frac{2k-3}{4k-2}} \qquad \text{and} \qquad \alpha_2=1-\sqrt{1-\frac{2k-3}{4k-2}}.
\]
Hence, we can express the $n$th coefficient of $G(z)$ as
\[
z[n]G(z)=\frac{(1-\frac{1}{\alpha_1})\cdot \alpha_1^{n+1}}{2-\frac{2k-3}{4k-2}\cdot\frac{2}{\alpha_1}}+\frac{(1-\frac{1}{\alpha_2})\cdot \alpha_2^{i+1}}{2-\frac{2k-3}{4k-2}\cdot\frac{2}{\alpha_2}}.
\]
We can now write down the smaller $4 \times 4$ system of equations. Let $|V| = |A_1| = |A_2| = \cdots = |A_{\frac{n+1}{2}}|$. From the equations of the signatures $h_1$ (from Lemma~\ref{lem:solalpha}), $h_2$ (from Lemma~\ref{lem:splitend}), and $h_3$ ($f_2$ from Lemma~\ref{lem:splitmiddle}), we get
\[
\begin{array}{ll}
 (|V|-2) \cdot \omega_1  +  |V| \cdot \omega_2  +  |V| \cdot \displaystyle \sum_{i=0}^{\frac{n+1}{2}-3}g_i\cdot\omega_{\frac{n+1}{2}}  & =  s \\
 |V| \cdot \omega_1  +  (|V|-(4k-2)) \cdot \omega_2  +  \displaystyle \sum_{i=0}^{\frac{n+1}{2}-3}(2k+1)\frac{n-1+2i}{2}\cdot g_i\cdot\omega_{\frac{n+1}{2}} & =  s \\ 
 |V| \cdot \omega_1  +  (|V|-(8k-4)) \cdot (\omega_2 + \left(\displaystyle \sum_{i=0}^{\frac{n+1}{2}-3}g_i - (4k-2)\right) \cdot \omega_{\frac{n+1}{2}})  & =  s \\
 |V|\cdot(\omega_1+\omega_2 + \displaystyle \sum_{i=0}^{\frac{n+1}{2}-3}g_i\cdot\omega_{\frac{n+1}{2}}) & = 1.
\end{array}
\]
Solving this gives
\[
s=\frac{(2kn+n-4)(\xi_n(4k-1)+(2k-1))}{(2kn+n-4)(\xi_n(4k-1)+(2k-1))+(4k-2)(1-\xi_n)},
\]
where
\[
\xi_n = 
\left(\alpha_1^{(n-1)/2} + \alpha_2^{(n-1)/2}\right)/4.
\]
\end{proof}

\section{Proofs of Results from Section~\ref{sec:apply}}

\subsection*{Lemma~\ref{lem:circlesandwich}}
\begin{proof}
Since $\chi_c(G) \leq r$ means there exist one $r' \leq r$ such that $G \rightarrow K_{r'}$, and since $K_r' \rightarrow K_r$ we have $G \rightarrow K_r$ and $K_2$ has a homomorphism to every graph that contains at least one edge so $K_2 \rightarrow G \rightarrow K_r$ and we can apply Lemma~\ref{lem:sandwich}. We also have that $C_{2k+1}$ has a homomorphism to each graph which contains an odd cycle with length at most $2k+1$. It is obvious that $C_{2k+1}$ has an homomorphism into a graph containing a cycle of length exactly $2k+1$. But we also know that $C_{2k+1} \rightarrow C_{2m+1}$ if $m \leq k$ so if $G$ contains an odd cycle of length at most $2k+1$ then we have $C_{2k+1} \rightarrow G \rightarrow K_r$. 
\end{proof}

\subsection*{Proposition~\ref{thmI}}
\begin{proof}
By Pan and Zhu we know the following for graphs $G$ that are
$K_4$-minor-free and integers $k\geq 1$:
\begin{itemize}
\item If G has odd girth at least $6k-1$ then $\chi_c(G)\leq 8k/(4k-1)$;
\item If G has odd girth at least $6k+1$ then $\chi_c(G)\leq (4k+1)/2k$;
\item If G has odd girth at least $6k+3$ then $\chi_c(G)\leq
(4k+3)/(2k+1)$.
\end{itemize}
The above, combined with Proposition~\ref{prop:4k+4}, can be used to
specify values on $s(K_2,G)$. We get that when the odd girth is at least
$6k-1$ then $s(K_2,G) \geq \frac{4k}{4k+1}$ and when the odd girth is at
least $6k+3$ then $s(K_2,G) \geq \frac{4k+2}{4k+3}$. For graphs with odd
girth $6k+1$ the result of Pan and Zhu give no other guarantee than that
a homomorphism exists to the cycle $C_{4k+1}$, which gives us no better
bound than for graphs with girth $6k-1$.
\end{proof}

%\bibliographystyleapp{abbrv}
%\bibliographyapp{wg,/home/chrba/bib/strings,/home/chrba/bib/abstracts,/home/chrba/bib/papers,/home/chrba/bib/books,/home/chrba/bib/xref,/home/petej/bib/extrapapers}

\end{document}